\documentclass[journal]{IEEEtran}

\usepackage{cite}
\usepackage[pdftex]{graphicx}
\usepackage[cmex10]{amsmath}
\usepackage{amssymb}
\usepackage{epstopdf}
\newtheorem{lemma}{Lemma}
\newtheorem{theorem}{Theorem}
\newtheorem{remark}{Remark}
\newtheorem{proposition}{Proposition}
\newtheorem{proof}{Proof}
\usepackage{flushend}
\usepackage{color}
\usepackage{algorithm}
\usepackage{algorithmic}

\usepackage{array}

\usepackage{bm}
\usepackage{amsfonts}
\usepackage{booktabs}

\usepackage[font=footnotesize]{subfig}

\begin{document}

\title{A High-resolution DOA Estimation Method with a Family of Nonconvex Penalties}

\author{Xiaohuan~Wu,~\IEEEmembership{Student~Member,~IEEE,}
        Wei-Ping~Zhu,~\IEEEmembership{Senior~Member,~IEEE,}
        and~Jun~Yan

\thanks{This work was supported by the National Natural Science Foundation of China under grant No. 61372122 and No. 61471205; the Innovation Program for Postgraduate in Jiangsu Province under grant No. KYLX 0813.}
\thanks{X. Wu and J. Yan are with the Key Lab of Broadband Wireless Communication and Sensor Network Technology, Nanjing University of Posts and Telecommunications, Nanjing, 210003, China (e-mail: 2013010101@njupt.edu.cn).}
\thanks{W.-P. Zhu is with the Department of Electrical and Computer Engineering, Concordia University, Montreal, Canada. He is also an Adjunct Professor with the School of Communication and Information Engineering, Nanjing University of Posts and Telecommunications, Nanjing, China (e-mail: weiping@ece.concordia.ca).}}

\maketitle

\begin{abstract}
The low-rank matrix reconstruction (LRMR) approach is widely used in direction-of-arrival (DOA) estimation. As the rank norm penalty in an LRMR is NP-hard to compute, the nuclear norm (or the trace norm for a positive semidefinite (PSD) matrix) has been often employed as a convex relaxation of the rank norm. However, solving a nuclear norm convex problem may lead to a suboptimal solution of the original rank norm problem. In this paper, we propose to apply a family of nonconvex penalties on the singular values of the covariance matrix as the sparsity metrics to approximate the rank norm. In particular, we formulate a nonconvex minimization problem and solve it by using a locally convergent iterative reweighted strategy in order to enhance the sparsity and resolution. The problem in each iteration is convex and hence can be solved by using the optimization toolbox. Convergence analysis shows that the new method is able to obtain a suboptimal solution. The connection between the proposed method and the sparse signal reconstruction (SSR) is explored showing that our method can be regarded as a sparsity-based method with the number of sampling grids approaching infinity. Two feasible implementation algorithms that are based on solving a duality problem and deducing a closed-form solution of the simplified problem are also provided for the convex problem at each iteration to expedite the convergence. Extensive simulation studies are conducted to show the superiority of the proposed methods.
\end{abstract}

\IEEEpeerreviewmaketitle

\section{Introduction}

As a fundamental problem in array signal processing, direction-of-arrival (DOA) estimation has been widely employed in many applications, i.e., radar, sonar, speech processing and wireless communications \cite{benesty2008microphone, wirelesscommunications, GuYujie2012TSP}. The study of DOA estimation problem has been a long history and many methods have been proposed in the past two decades, including some well known conventional methods such as Capon method and the subspace-based methods, e.g., MUSIC (see \cite{krim1996two} for a detailed review). MUSIC is equivalent to a large sample realization of the maximum likelihood (ML) method when the signals are uncorrelated, and has a super-resolution compared to Capon method under certain conditions \cite{johnson1992array}. However, MUSIC requires the number of sources as \emph{a priori} and a one-dimensional ($1$-D) searching procedure which is computationally expensive. Moveover, MUSIC faces difficulties in the case of sparse linear array (SLA), especially when the number of sources is greater than that of sensors. Although the root-MUSIC has been proposed for efficiency consideration, it only can be used for the uniform linear array (ULA).

Compressive sensing \cite{Donoho2006} is a technique of reconstructing a high dimensional signal from fewer samples and has been introduced into the DOA estimation area, resulting in a number of sparsity-based methods for DOA estimation \cite{malioutov2005sparse, hyder2010direction, liuzhangmeng2012SBL, liu2013CMSR, L1SRACV2011TSP}. The sparsity-based methods exhibit several advantages over subspace-based and the ML methods such as robustness to noise, no requirement for source number, and improved resolution. However, the sparsity-based methods require the DOAs of interest to be sparse in the whole angle space. To this end, the angle space has to be discretized into a finite set of angle grids and the DOAs of interest are assumed to exactly lie on the grids. However, in reality, the DOAs could lie in the continuous infinite set of angles, and hence the assumption holds only when the size of the set tends to infinity, which results in an unacceptable computational cost. Moreover, the discretization strategy may degenerate the performance of the sparsity-based methods as there is often an unavoidable bias between the true DOAs and the predefined grids, which can be interpreted as the basis mismatch issue. To address this issue, several modified off-grid sparsity-based methods have been proposed \cite{mySJSBL, yangzai2013off-grid, Zhuhao2011STLSoriginaloffgrid, TanZhao_offgrid2014TSP}. However, these methods are the mitigation measures only, since the discretization still exists which does not eliminate the effect of basis mismatch but bring complexity issue.

Recently, the atomic norm is introduced as a mathematical theory by V. Chandrasekaran \emph{et al.} \cite{atomicnorm2012}, and then extended for line spectral estimation by Tang in \cite{Tang2013offthegrid}. Since DOA estimation problem actually equates to the frequency recovery problem in the presence of multiple measurement vectors (MMVs) which share the same frequency components, the atomic norm minimization (ANM) technique can be directly applied to the DOA estimation, as explored by Yang in \cite{yangzai2015GLS}. The ANM can be referred to as a gridless method since it views the DOA estimation as a sparse reconstruction problem with a continuous infinite dictionary and recovers the DOAs by solving a semidefinite programming (SDP) problem. It is shown that the DOAs can be exactly recovered with a high probability provided that they are appropriately separated. Note that since the discretization is not required, ANM is able to eliminate the effect of basis mismatch. However, the DOA separation condition limits the estimation precision of ANM. Furthermore, ANM fails to give a satisfying performance in the moderate range of the signal-to-noise ratio (SNR).
The structured covariance matrix reconstruction (SCMR), also referred to as the covariance matching technique, is another class of gridless methods \cite{ITAM1988ICASSP, COMETottersten1998covariance, AML1999LiHBTSP, SPICE2011TSP, yangzai2014SPA, Liyuanxin2016TSP}, which explores the Hermitian and Toeplitz structural information of the covariance matrix as \emph{a priori}. In particular, the existing covariance matching criteria (CMC) can be divided into three categories as shown in Table \ref{table CMC}.\footnote{In Table \ref{table CMC}, $\hat{\bm{R}}$ denotes the sample covariance matrix and $\bm{R}$ denotes the unknown covariance matrix with a Toeplitz structure to be estimated.}
All these CMC-based methods are statistically consistent in the number of snapshots. In comparison, CMC$1$ is shown to be inconsistent in SNR since it makes no effort to ensure appropriate subspace approximation while CMC$2$ and CMC$3$ yield consistent DOA estimates as the SNR increases. On the other hand, CMC$2$ exhibits better estimation performance than CMC$3$ when SNR is low or moderate while CMC$3$ is applicable to the correlated signals and single-snapshot cases compared to CMC$2$ \cite{yangzai2014SPA}.

\begin{table}[!t]
\renewcommand{\arraystretch}{1.3}
\caption{The covariance matching criteria (CMC) and the corresponding representative methods. }
\label{table CMC}
\centering
\begin{tabular}{ccc}\toprule
Category & Formulation & Representative Methods \\
\midrule
CMC$1$ & $\Big\| \hat{\bm{R}} - \bm{R} \Big\|_F^2$ &  ITAM \cite{ITAM1988ICASSP}, method in \cite{Liyuanxin2016TSP} \\
\specialrule{0em}{2pt}{2pt}
CMC$2$ & $\Big\| \hat{\bm{R}}^{-\frac{1}{2}} (\hat{\bm{R}} - \bm{R}) \hat{\bm{R}}^{-\frac{1}{2}} \Big\|_F^2$ & COMET \cite{COMETottersten1998covariance}, AML \cite{AML1999LiHBTSP} \\
\specialrule{0em}{2pt}{2pt}
CMC$3$ & $\Big\| \bm{R}^{-\frac{1}{2}} (\hat{\bm{R}} - \bm{R}) \hat{\bm{R}}^{-\frac{1}{2}} \Big\|_F^2$ & SPICE \cite{SPICE2011TSP}, SPA \cite{yangzai2014SPA} \\
\bottomrule
\end{tabular}
\end{table}

Recently, we have proposed a gridless method for DOA estimation named covariance matrix reconstruction approach (CMRA) based on CMC$2$ in \cite{my2016ICASSP} and \cite{my2016TVT}. CMRA is applicable to the ULA and SLA and can be regarded as the gridless version of the sparsity-based method. Nevertheless, since CMRA approximates the rank norm by the trace norm, which is a loose approximation of the rank norm, there may exist a large gap between the solutions by using the two norms \cite{GapBetweenTraceAndRankNorm2011ISIT}. Such a difference is similar to that between the $\ell_1$ and $\ell_0$ norms. In particular, for reconstructing a sparse vector, to promote sparsity to the greatest extent possible, the natural choice of the sparsity metric should be $\ell_0$ norm, which, however, is NP-hard to compute. A practical choice of the sparsity metric is the $\ell_1$ norm which is a convex relaxation but a loose approximation of the $\ell_0$ norm. The gap between the two norms can be mitigated by the nonconvex surrogates, e.g., the $\ell_p (0<p<1)$ norm, Logarithm, etc., which are well-studied in literature \cite{L_pNorm2014Mohammadi, Logarithm2003FazelandBoyd, Logarithm2012Mohan, LaplaceNorm2009Trzasko}. Hence, even though CMRA is able to eliminate the basis mismatch effect and give satisfactory performance in DOA estimation, we believe its ability can be further improved by employing the nonconvex surrogates, which, to the best of our knowledge, have not been studied in literature.

In this paper, motivated by the relationship between the $\ell_1$ and $\ell_0$ norms, we introduce a family of nonconvex continuous surrogate functions as the sparsity metrics rather than the convex trace norm in CMRA and propose an iteratively reweighted version of CMRA algorithm, called improved CMRA (ICMRA) to solve the nonconvex problem. Further analysis shows that ICMRA serves as CMRA in the first iteration and a reweighted CMRA in each of the following iterations with the weights being based on the latest estimate. A convergence analysis is also provided to verify that ICMRA is able to obtain at least a suboptimal solution. We then explore the connection between the ICMRA and the sparsity-based methods, as well as the atomic norm, showing that ICMRA is able to enhance the sparsity. Two feasible implementation algorithms are also provided to speed up the convergence. One is to solve its duality problem since by which a faster convergence is observed. The other one gives a closed-form solution by simplying the constraint. Extensive numerical simulations are carried out to verify the performance of our proposed methods.

Notations used in this paper are as follows. $\mathbb{C}$ denotes the set of complex numbers. For a matrix $\bm{A}$, $\bm{A}^T$, $\bar{\bm{A}}$ and $\bm{A}^H$ denote the transpose, conjugate and conjugate transpose of $\bm{A}$, respectively. $\|\bm{A}\|_F$ denotes the Frobenius norm of $\bm{A}$, vec$(\bm{A})$ denotes the vectorization operator that stacks matrix $\bm{A}$ column by column, $\bm{A}\odot\bm{B}$ is the Khatri-Rao product of matrices $\bm{A}$ and $\bm{B}$, and tr$(\bullet)$ and rank$(\bullet)$ are the trace and rank operators, respectively. 
$\bm{A}\geq \bm{0}$ means that matrix $\bm{A}$ is positive semidefinite. $\bm{\sigma}[\bm{A}]$ and $\sigma_i[\bm{A}]$ denote the singular value and the $i$-th singular value of $\bm{A}$, respectively. For a vector $\bm{x}$, $\bm{x}\succeq \bm{0}$ means that every entry of $\bm{x}$ is nonnegative, $\|\bm{x}\|_2$ denotes the $\ell_2$ norm of $\bm{x}$, diag$(\bm{x})$ denotes a diagonal matrix with the diagonal entries being the entries of vector $\bm{x}$ in turn. $\text{Re}(\bullet)$ and $\text{Im}(\bullet)$ stand for the real and the imaginary parts of a complex variable, respectively.

The rest of the paper is organized as follows. Section \ref{sec preliminaries} revisits the signal model, the CMRA method and the atomic norm as the preliminaries. Section \ref{sec reweighted CMRA} presents a family of nonconvex sparsity metrics and then introduces the ICMRA method. A convergence analysis is also provided at the end of this section. Section \ref{sec connection to SSR and atomic norm} reveals the relationship between ICMRA and some of the existing methods. Section \ref{sec computationally efficient implementations} provides two feasible implementation algorithms. Simulations are carried out in Section \ref{sec simulation} to demonstrate the performance of our methods. Finally, Section \ref{sec conclusions} concludes the whole paper.

\section{Preliminaries}
\label{sec preliminaries}

\subsection{Signal Model}

For better illustration, we denote $\bm{\Omega} = \{\Omega_1,\cdots,\Omega_M\} \subseteq \{1,\cdots,N\}$ as the sensor index set of a linear array. In particular, the ULA has $\bm{\Omega} = \{1,\cdots,N\}$ and the SLA has $\bm{\Omega}\subset \{1,\cdots,N\}$ with $\Omega_1 = 1$ and $\Omega_M = N$, where $M$ is the number of sensors.\footnote{The SLA considered in this paper only involves the array whose coarray is an $N$-element ULA.} Note that the ULA can be considered as a special case of SLA, hence we use the SLA in the following unless otherwise stated.

Assume $K$ narrowband far-field signals impinge onto an SLA with $M$ sensors from directions of $\bm{\theta} = \{\theta_1,\cdots,\theta_K\}$. The array output with respect to $L$ snapshots is,
\begin{equation}\label{eq signal model}
  \bm{X}_{\bm{\Omega}} = \bm{A}_{\bm{\Omega}} \bm{S} + \bm{N},
\end{equation}
where $\bm{X}_{\bm{\Omega}}\in \mathbb{C}^{M\times L}$, $\bm{A}_{\bm{\Omega}} = [\bm{a}_{\bm{\Omega}}(\theta_1),\cdots,\bm{a}_{\bm{\Omega}}(\theta_K)]\in \mathbb{C}^{M\times K}$ and $\bm{S}\in \mathbb{C}^{K\times L}$ denote the array output, the manifold matrix and the waveform of the impinging signals, respectively. $\bm{N}\in \mathbb{C}^{M\times L}$ is the white Gaussian noise with zero mean. The steering vector $\bm{a}_{\bm{\Omega}}(\theta_k) = [e^{j\pi (\Omega_1-1) \sin \theta_k}, \cdots, e^{j\pi (\Omega_M-1) \sin \theta_k}]^T$ contains the DOA information which needs to be determined.

\subsection{The CMRA Method}

The CMRA aims to first reconstruct the covariance matrix of the coarray of the SLA according to the output. In particular, 
let $\bm{\Gamma} \in \{0,1\}^{M\times N}$ be a selection matrix such that the $m$-th row of $\bm{\Gamma}$ contains all zeros but a single $1$ at the $\Omega_m$-th position. Then the covariance matrix of $\bm{X}_{\bm{\Omega}}$ can be written as,
\begin{equation}\label{eq CMRA1}
\begin{split}
  \bm{R_{\Omega}} &= \bm{\Gamma} T(\bm{u}) \bm{\Gamma}^T + \sigma \bm{I}\\
                  &= T_{\bm{\Omega}}(\bm{u}) + \sigma \bm{I},
\end{split}
\end{equation}
where $\sigma$ denotes the noise power and $T_{\bm{\Omega}}(\bm{u}) \triangleq \bm{\Gamma} T(\bm{u}) \bm{\Gamma}^T$ in which $T(\bm{u})$ is a Toeplitz Hermitian matrix with $\bm{u}=[u_1,\cdots,u_N]^T$ being the first column of $T(\bm{u})$. In practical applications, $\bm{R_{\Omega}}$ is approximated by the sample covariance matrix with $L$ snapshots as,
\begin{equation}\label{eq covariance}
  \hat{\bm{R}}_{\bm{\Omega}} = \frac{1}{L}\bm{X}_{\bm{\Omega}}\bm{X}^H_{\bm{\Omega}}.
\end{equation}
The error between $\hat{\bm{R}}_{\bm{\Omega}}$ and $\bm{R_{\Omega}}$ is defined as $\bm{E}_{\bm{\Omega}} = \hat{\bm{R}}_{\bm{\Omega}} - \bm{R_{\Omega}}$ which has the following property,
\begin{equation}\label{eq CMRA2}
  \big\| \bm{Q} \text{vec}(\bm{E_{\Omega}}) \big\|_2^2 \sim \text{As}\chi^2(M^2),
\end{equation}
where $\bm{Q}=\sqrt{L} \bm{R}_{\bm{\Omega}}^{-\frac{T}{2}} \otimes \bm{R}_{\bm{\Omega}}^{-\frac{1}{2}}$, and $\text{As}\chi^2(M^2)$ denotes the asymptotic chi-square distribution with $M^2$ degrees of freedom (see \cite{my2016ICASSP}, \cite{Liu2013CV-RVM} for more details). As a result, CMRA formulates the following low-rank matrix reconstruction (LRMR) model for $T(\bm{u})$,
\begin{equation}\label{eq CMRA rank model}
  \min_{\bm{u}}\ \text{rank}[T(\bm{u})]\quad \text{s.t.}\ \big\| \bm{Q} \text{vec}(\bm{E}_{\bm{\Omega}}) \big\|_2^2 \leq \beta^2,\ T(\bm{u}) \geq \bm{0},
\end{equation}
which is then relaxed into the following convex optimization model by replacing the rank norm with the nuclear norm, or equivalently, the trace norm for a positive semidefinite (PSD) matrix,
\begin{equation}\label{eq CMRA3}
  \min_{\bm{u}}\ \text{tr}[T(\bm{u})]\quad \text{s.t.}\ \big\| \bm{Q} \text{vec}(\bm{E}_{\bm{\Omega}}) \big\|_2^2 \leq \beta^2,\ T(\bm{u}) \geq \bm{0},
\end{equation}
where $\beta$ can be easily determined by using the property of the chi-square distribution in (\ref{eq CMRA2}).\footnote{In this paper, parameter $\beta$ in CMRA and the proposed ICMRA is calculated using MATLAB routine \texttt{chi2inv}$(1-p, M^2)$, where $p$ is set to 0.001 in general.} For
simplicity, we assume the noise power $\sigma$ can be estimated as the smallest eigenvalue of $\hat{\bm{R}}_{\bm{\Omega}}$.\footnote{This assumption, which is a compromise due to the lack of knowledge of the number of signals, is reasonable with moderate or large number of snapshots since the error-suppression criterion in (\ref{eq CMRA3}) can tolerate well this underestimated error \cite{L1SRACV2011TSP}.} After obtaining $T(\bm{u})$ by solving (\ref{eq CMRA3}), the DOAs $\bm{\theta}$ can be easily determined by using the subspace-based methods or the Vandermonde decomposition lemma (see \cite{yangzai2014SPA} for more details).

\begin{remark}\label{Remark Rfull}
  It should be noted that, in the case of using an SLA, if the number of impinged signals is equal to or larger than that of sensors, the space of the impinged signals will exceed the space generated by the array. As a consequence, estimating the noise power as the smallest eigenvalue of $\hat{\bm{R}}_{\bm{\Omega}}$ is impractical. To handle this special case, we should first reconstruct the whole space spanned by the coarray, where the space of the impinged signals not fully occupy. To do this, the Hermitian Toeplitz structure of the covariance matrix can be used. In particular, let $\hat{\bm{R}}^{\text{full}} = \bm{\Gamma}^T \hat{\bm{R}}_{\bm{\Omega}} \bm{\Gamma}$ and replace its zero elements with the mean of the nonzero entries in the same diagonals. In this way, the impinged signals will fall into the space spanned by the coarray and the noise power can be obtained as the smallest eigenvalue of $\hat{\bm{R}}^{\text{full}}$.
\end{remark}

\subsection{Atomic Norm}

The concept of atomic norm was first introduced in \cite{atomicnorm2012}, which generalizes many norms such as $\ell_1$ norm and the nuclear norm. Considering the DOA estimation using a ULA, let
\begin{equation}\label{eq atomic set}
  \mathcal{A} = \Big\{\frac{1}{N}\bm{a}(\theta)\bm{a}^H(\theta):\theta\in [-90^{\circ}, 90^{\circ}]\Big\},
\end{equation}
which is a set of unit-norm rank-one matrices. The atomic $\ell_0$ norm of $T(\bm{u})$ is defined as the smallest number of atoms in $\mathcal{A}$ composing $T(\bm{u})$, which is shown as,
\begin{equation}\label{eq atomic l0 norm}
  \begin{split}
  &\quad \| T(\bm{u}) \|_{\mathcal{A},0}\\
   &= \inf_{c_k> 0} \left\{ \mathcal{K} : T(\bm{u}) = \sum_{k=1}^{\mathcal{K}} c_k\bm{B}(\theta_k),\ \bm{B}(\theta_k)\in \mathcal{A} \right\}\\
  &= \text{rank}[T(\bm{u})].
\end{split}
\end{equation}
Hence, model (\ref{eq CMRA rank model}) is equivalent to,
\begin{equation}\label{eq atomic l0 model}
  \min_{\bm{u}}\ \|T(\bm{u})\|_{\mathcal{A},0} \quad \text{s.t.}\ \big\| \bm{Q} \text{vec}(\bm{E}) \big\|_2^2 \leq \beta^2,\ T(\bm{u}) \geq \bm{0}.
\end{equation}
It is easy to see that although the atomic $\ell_0$ norm directly enhances sparsity, it is nonconvex and NP-hard to compute. Hence the relaxation of $\| T(\bm{u}) \|_{\mathcal{A},0}$, named the atomic norm, is introduced as,
\begin{equation}\label{eq atomic norm}
  \begin{split}
  &\quad \| T(\bm{u}) \|_{\mathcal{A}}\\
   &= \inf_{c_k> 0} \left\{ \sum_k c_k : T(\bm{u}) = \sum_{k=1}^{\mathcal{K}} c_k\bm{B}(\theta_k),\ \bm{B}(\theta_k)\in \mathcal{A} \right\}\\
  &= \sum_k N p_k\\
  &= \text{tr}[T(\bm{u})],
\end{split}
\end{equation}
which indicates that model (\ref{eq CMRA3}) is equivalent to
\begin{equation}\label{eq atomic model}
  \min_{\bm{u}}\ \|T(\bm{u})\|_{\mathcal{A}} \quad \text{s.t.}\ \big\| \bm{Q} \text{vec}(\bm{E}) \big\|_2^2 \leq \beta^2,\ T(\bm{u}) \geq \bm{0}.
\end{equation}

It should be mentioned that, the solution obtained by solving (\ref{eq CMRA3}) or (\ref{eq atomic model}) is usually suboptimal to the corresponding LRMR problem since the trace norm is a loose approximation of the rank norm, i.e., there still exists a noticeable gap between the solutions of the two norms. To achieve a better approximation of the rank norm, the next section, we propose an iterative reweighted method by introducing a family of nonconvex penalties as the sparsity metrics rather than the trace norm.

\section{Improved CMRA}
\label{sec reweighted CMRA}

\subsection{The Novel Sparsity Metrics}

\begin{figure}[!t]
\centering
\includegraphics[width=3in]{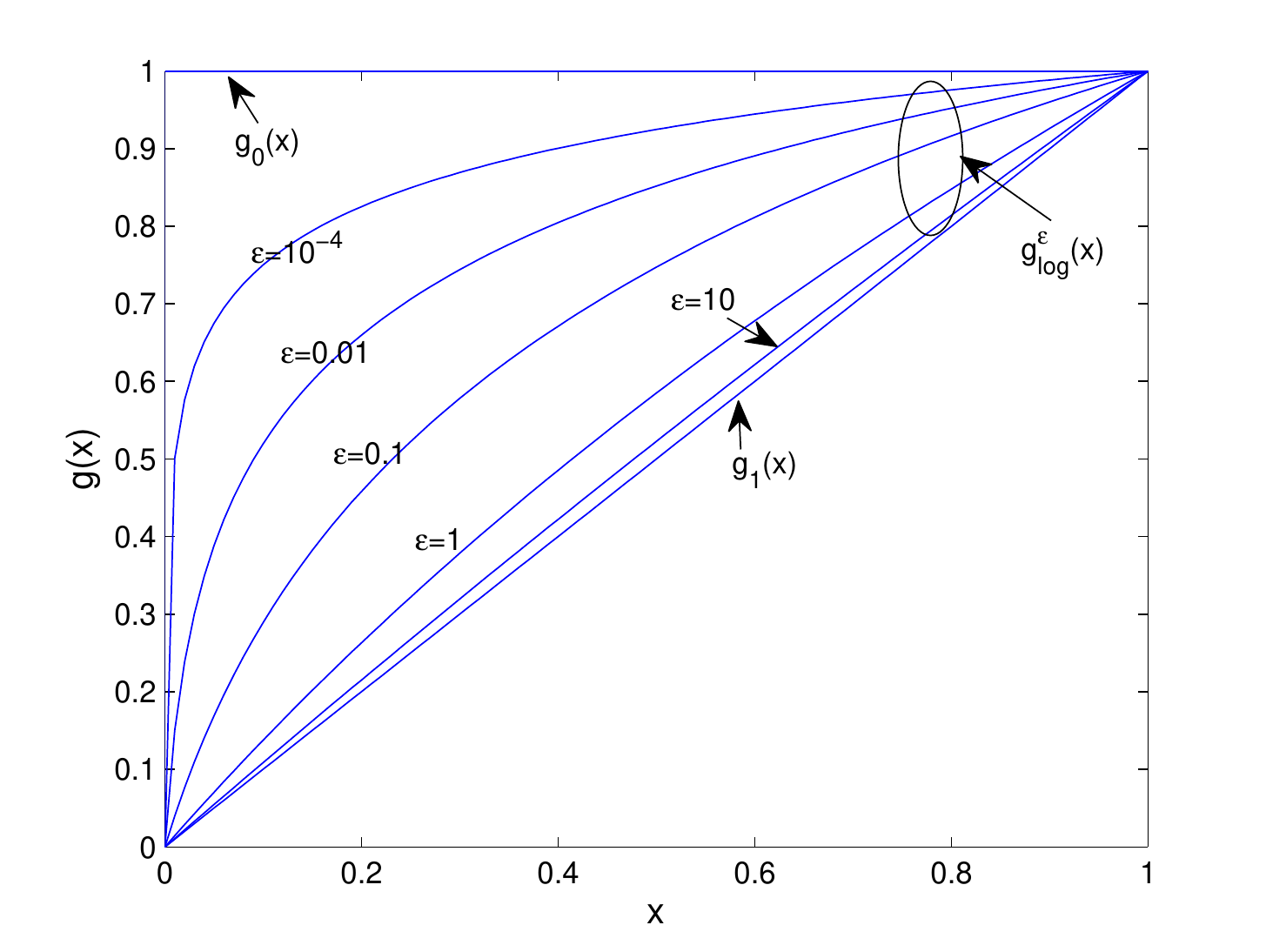}
\caption{Illustration of $g_0(x)$, $g_1(x)$ and $g_{\text{log}}^{\epsilon}(x)$.}
\label{Fig illustration of g0 g1}
\end{figure}

It is easy to see that the rank norm in model (\ref{eq CMRA rank model}) is nonconvex and challenging to solve, while the trace norm, which is utilized in the CMRA method, is computable and the best convex approximation of the rank norm but has the worst fitting. Hence it is essential to find a better approximation of the rank norm but still with a low computational complexity. In particular, for a matrix $\bm{X} \in \mathbb{C}^{N\times N}$, the rank and nuclear norm can be represented as
\begin{equation}\label{eq rank trace}
\begin{split}
  &\text{rank}[\bm{X}] = \min_i \sigma_i(\bm{X})>0 = \sum_i g_0[\sigma_i(\bm{X})];\\
  &\|\bm{X}\|_* = \sum_i \sigma_i(\bm{X}) = \sum_i g_1[\sigma_i(\bm{X})],
\end{split}
\end{equation}
respectively, where we have assumed $\sigma_1[\bm{X}]\geq \sigma_2[\bm{X}]\geq \cdots \geq \sigma_N[\bm{X}] \geq 0$, and
\begin{equation}\label{eq RCMRA1}
  g_0(x) = \Big\{ \begin{matrix} 1&x>0\\0&x=0 \end{matrix}\ ; \quad g_1(x) = x,\ x\geq 0.
\end{equation}
Hence finding an approximation of the rank norm is equivalent to find a nonconvex sparsity metric $g(x)$ which bridges the gap between $g_0(x)$ and $g_1(x)$  (see Fig. \ref{Fig illustration of g0 g1}).
To the best of our knowledge, there exist several nonconvex penalties in literature including Logarithm \cite{Logarithm2003FazelandBoyd}, $\ell_p$ norm \cite{L_pNorm2014Mohammadi} and Laplace \cite{LaplaceNorm2009Trzasko} which we denote as $g^{\epsilon}_{\text{log}}(x)$, $g^{\epsilon}_{\ell_p}(x)$ and $g^{\epsilon}_{\text{lap}}(x)$, respectively, where $\epsilon$ denotes the trade-off parameter. To illustrate these penalties more precisely, we enumerate of them in Table \ref{table penalties} for comparison and show the curve of $g^{\epsilon}_{\text{log}}(x)$ with respect to different $\epsilon$ as an example in Fig. \ref{Fig illustration of g0 g1},\footnote{$g^{\epsilon}_{\text{log}}(x)$ is translated and scaled such that it equals $0$ and $1$ at $x = 0$ and $1$ respectively for better illustration.} from which we can see that $g^{\epsilon}_{\text{log}}(x)$ approaches $g_1(x)$ with a large value of $\epsilon$ and gets close to $g_0(x)$ when $\epsilon\rightarrow 0$. Note that $g^{\epsilon}_{\ell_p}(x)$ and $g^{\epsilon}_{\text{lap}}(x)$ have similar properties which are omitted for brevity. In fact, these functions $g$ can be employed as the nonconvex penalties since they satisfy the following properties:

\textbf{P1}: $g$ is concave, monotonically increasing on $[0,+\infty)$.

\textbf{P2}: $g$ is continuous but possibly nonsmooth on $[0,+\infty)$.

\textbf{P3}: after being translated and scaled, $g$ approaches $g_0$ when $\epsilon\rightarrow 0$ and $g_1$ when $\epsilon$ is large.

Some other functions satisfying the aforementioned properties have also been proposed as the penalties, e.g., smoothly clipped absolute deviation (SCAD) \cite{SCAD2001Fan} and minimax concave penalty (MCP) \cite{MCP2010Zhang}, which, however, are omitted in this paper because they usually fail to give a satisfactory performance when employed in an LRMR problem \cite{IRNN2016TIP}.
\begin{table}[!t]
\renewcommand{\arraystretch}{1.3}
\caption{Some nonconvex penalties of $g_0(x)$ and their gradients. }
\label{table penalties}
\centering
\begin{tabular}{ccc}
\toprule
\text{Penalty} & $g^{\epsilon}(x),x\geq0,\epsilon>0$ & $\text{gradient}\ \nabla g^{\epsilon}(x)$ \\
\midrule
\text{Logarithm} & $\ln(x+\epsilon)$ & $\frac{1}{x+\epsilon}$ \\
\specialrule{0em}{1pt}{1pt}
$\ell_p$ \text{norm} & $x^{\epsilon}$ & $\epsilon x^{\epsilon-1}$ \\
\specialrule{0em}{1pt}{1pt}
\text{Laplace} & $1-e^{-\frac{x}{\epsilon}}$ & $\frac{1}{\epsilon}e^{-\frac{x}{\epsilon}}$ \\
\bottomrule
\end{tabular}
\end{table}

Motivated by (\ref{eq rank trace}), the rank norm model (\ref{eq CMRA rank model}) and the trace norm model (\ref{eq CMRA3}) can be rewritten as,
\begin{equation}\label{eq rank and trace}
\begin{split}
  &\min_{\bm{u}}\ \mathcal{G}[T(\bm{u})] \\ &\quad\text{s.t.}\ \big\| \bm{Q} \text{vec}(\bm{E}_{\bm{\Omega}}) \big\|_2^2 \leq \beta^2,\ T(\bm{u}) \geq \bm{0},
\end{split}
\end{equation}
where $\mathcal{G}[T(\bm{u})] = \sum_i g[\sigma_i(T(\bm{u}))]$ with $g[\sigma_i(T(\bm{u}))]$ being $g_0 [\sigma_i(T(\bm{u}))]$ for model (\ref{eq CMRA rank model}) or $g_1[\sigma_i(T(\bm{u}))]$ for model (\ref{eq CMRA3}).

Further inspired by the link between $g^{\epsilon}(x)$ and $g_0(x)$ or $g_1(x)$ above, we propose the following general nonconvex optimization model,
\begin{equation}\label{eq RCMRA2}
\begin{split}
  &\min_{\bm{u}}\ \mathcal{G}^{\epsilon}[T(\bm{u})] \\ &\quad\text{s.t.}\ \big\| \bm{Q} \text{vec}(\bm{E}_{\bm{\Omega}}) \big\|_2^2 \leq \beta^2,\ T(\bm{u}) \geq \bm{0},
\end{split}
\end{equation}
where $\mathcal{G}^{\epsilon}[T(\bm{u})]= h^{\epsilon}\big[\bm{\sigma}[T(\bm{u})]\big] =\sum_i g^{\epsilon}\big[\sigma_i[T(\bm{u})]\big]$. Intuitively, we expect the new nonconvex model to bridge the gap between models (\ref{eq CMRA rank model}) and (\ref{eq CMRA3}) when $\epsilon$ varies from a large number to zero.

\subsection{An Iteratively Reweighted Algorithm}
\label{sec reweighted CMRA B}
Since model (\ref{eq RCMRA2}) is nonconvex and no efficient algorithms can guarantee to obtain the global minimum, we use the majorization-maximization (MM) method to obtain a suboptimal solution instead. The MM method is an iterative approach and the cost function is replaced by its tangent plane in each iteration. In particular, denote $\bm{u}_j$ as the optimization variable of the $j$-th iteration. Since $\mathcal{G}^{\epsilon}[T(\bm{u})]=\sum_i g^{\epsilon}\big[\sigma_i[T(\bm{u})]\big]$ is concave, we have,
\begin{equation}\label{eq RCMRA3}
  \mathcal{G}^{\epsilon}[T(\bm{u})] \leq \mathcal{G}^{\epsilon}[T(\bm{u}_j)] + \text{tr}\big[\nabla\mathcal{G}^{\epsilon}[T(\bm{u}_j)]T(\bm{u}-\bm{u}_j)\big],
\end{equation}
As a result, the optimization problem at the $(j+1)$-th iteration becomes,
\begin{equation}\label{eq RCMRA4}
\begin{split}
    &\min_{\bm{u}}\ \text{tr}\big[\bm{W}_j^{\epsilon}T(\bm{u})\big] \\ &\quad\text{s.t.}\ \big\| \bm{Q} \text{vec}(\bm{E}_{\bm{\Omega}}) \big\|_2^2 \leq \beta^2,\ T(\bm{u}) \geq \bm{0},
\end{split}
\end{equation}
where $\bm{W}_j^{\epsilon} \triangleq \nabla\mathcal{G}^{\epsilon}[T(\bm{u}_j)]$. To calculate $\bm{W}_j^{\epsilon}$, we use the following proposition,
\begin{proposition}[\cite{ICRA2014Mohammadi}]
  Suppose that $\mathcal{G}^{\epsilon}(\bm{X})$ is represented as $\mathcal{G}^{\epsilon}(\bm{X}) = h^{\epsilon}(\bm{\sigma}(\bm{X})) = \sum_i g^{\epsilon}(\sigma_i(\bm{X}))$, where $\bm{X}\geq \bm{0}$ with the singular value decomposition (SVD) $\bm{X} = \bm{D}\text{diag}[\bm{\sigma}(\bm{X})]\bm{D}^T$. Denote two mappings $h^{\epsilon}$ and $f^{\epsilon}$ which are differentiable and concave. Then the gradient of $\mathcal{G}^{\epsilon}(\bm{X})$ at $\bm{X}$ is
  \begin{equation}\label{eq gradient of G}
    \nabla \mathcal{G}^{\epsilon}(\bm{X}) = \bm{D} \text{diag}(\bm{\eta}) \bm{D}^T,
  \end{equation}
  where $\bm{\eta} = \nabla h^{\epsilon}(\bm{\sigma}(\bm{X}))$ denotes the gradient of $h$ at $\bm{\sigma}(\bm{X})$.
\end{proposition}
For simplicity, we denote $\bm{\sigma}^j = \bm{\sigma}[T(\bm{u}_j)]$. Based on the proposition above, we have,
\begin{equation}\label{eq RCMRA5}
  \bm{W}_j^{\epsilon} = \bm{U}_j \text{diag}\big[\nabla h^{\epsilon}(\bm{\sigma}^j)\big] \bm{U}_j^H,
\end{equation}
where $T(\bm{u}_j) = \bm{U}_j \text{diag}(\bm{\sigma}^j) \bm{U}_j^H$ is the SVD of $T(\bm{u}_j)$. It should be noted that, when the Logarithm or the $\ell_p$ norm penalty is employed, equation (\ref{eq RCMRA5}) can be accelerated as $\bm{W}_j^{\epsilon} = (T(\bm{u}_j) + \epsilon \bm{I})^{-1}$ or $\bm{W}_j^{\epsilon} = \epsilon [T(\bm{u}_j)]^{\epsilon-1}$.

As mentioned before, $\epsilon$ controls the relationship between $g^{\epsilon}(\sigma)$ and $g_0(\sigma)$ or $g_1(\sigma)$. In particular, with a small value of $\epsilon$, $g^{\epsilon}(\sigma)$ approaches $g_0(\sigma)$ but suffers from many local minima, whereas a large value of $\epsilon$ pushes $g^{\epsilon}(\sigma)$ toward the convex $g_1(\sigma)$, which however has the worst fitting to $g_0(\sigma)$. Consequently, we should start with a large value of $\epsilon$, and gradually decrease it during the iteration to reduce the risk of getting trapped in local minima. Moreover, in each iteration, we also define the weight $\bm{W}_j^{\epsilon}$ using the latest solution for avoiding local minima.

After obtaining $T(\bm{u})$, similar to CMRA, the DOAs $\bm{\theta}$ can be easily determined by using the subspace-based methods or the Vandermonde decomposition lemma (see \cite{yangzai2014SPA} for more details).
Before closing this subsection, we summarize the proposed ICMRA in Algorithm \ref{algorithm ICMRA}.
\begin{algorithm}[h]
\caption{ICMRA}
\label{algorithm ICMRA}
\begin{algorithmic}
\REQUIRE measurement data $\bm{X}_{\bm{\Omega}}$, $\beta$.
\STATE \textbf{Initialization}: $j=0, \bm{u}_0 = \bm{0}, \epsilon$.\\
\REPEAT
\STATE $\quad$Update the weight $\bm{W}_{j}^{\epsilon}$ by equation (\ref{eq RCMRA5});
\STATE $\quad$Update $\bm{u}_{j+1}$ by solving problem (\ref{eq RCMRA4});
\STATE $\quad$ $\epsilon = \frac{\epsilon}{\delta}\ (\delta>1)$,
\STATE $\quad$ $j=j+1$,
\STATE \textbf{end}
\UNTIL{Convergence}
\ENSURE $\hat{\bm{\theta}}$ by using the Vandermonde decomposition lemma.
\end{algorithmic}
\end{algorithm}
\begin{remark}
It should be noted that, when starting with $\bm{u}_0 = \bm{0}$, the weight $\bm{W}_0^{\epsilon} = c\bm{I}$ where $c$ is a positive nonzero constant. Hence the first iteration of ICMRA reduces to the CMRA method. From the second iteration, the weight $\bm{W}_j^{\epsilon}$ is determined based on the solution of the previous iteration and thus a reweighted CMRA is carried out in each iteration.
\end{remark}

\begin{remark}
The ICMRA is able to enhance the sparsity and give a better performance than CMRA does, which can be justified as follows.

The problem (\ref{eq RCMRA4}) can be written as (the two constraints in (\ref{eq RCMRA4}) are omitted for brevity),
\begin{equation}\label{eq RCMRA6}
\begin{split}
    &\quad \min_{\bm{u}}\ \text{tr}\big[\bm{W}_j^{\epsilon}T(\bm{u})\big]\\
    &= \min_{p_k,\theta_k}\ \text{tr}[\bm{W}_j^{\epsilon} \sum_k p_k \bm{a}(\theta_k)\bm{a}^H(\theta_k)]\\
    &= \min_{p_k,\theta_k}\ \sum_k \bm{a}^H(\theta_k) \bm{W}_j^{\epsilon} \bm{a}(\theta_k) p_k\\
    &= \min_{p_k,\theta_k}\ \sum_k \omega_k^{-1}p_k \quad \text{s.t.}\quad \omega_k = \frac{1}{\bm{a}^H(\theta_k) \bm{W}_j^{\epsilon} \bm{a}(\theta_k)},
\end{split}
\end{equation}
where $\bm{a}(\theta_k)$ denotes the steering vector of the coarray of the original array with a signal impinging from direction of $\theta_k$. Recall that $\bm{W}_j^{\epsilon} = (T(\bm{u}_j) + \epsilon \bm{I})^{-1}$ when the Logarithm penalty is used, hence $\omega_k$ can be regarded as the power spectrum of the Capon's beamforming if $T(\bm{u}_j)$ is interpreted as the covariance matrix of the noiseless array output and $\epsilon$ as the noise power. Therefore, the weights $\{\omega_k\}$ lead to finer details of the power spectrum in the current iteration and hence enhance the sparsity \cite{yang2016RAM}. Furthermore, since $g_{\ell_p}^{\epsilon}(\sigma^j)$ and $g_{\text{lap}}^{\epsilon}(\sigma^j)$ have similar sparsity enhancing properties to $g_{\text{log}}^{\epsilon}(\sigma^j)$ as shown in Fig. \ref{Fig illustration of g0 g1}, they can also be used for performance improvement, which will be shown in simulations.
\end{remark}

\subsection{Convergence Analysis}

In this section, we give the convergence analysis of ICMRA for (\ref{eq RCMRA2}) as the following theorem.
\begin{theorem}\label{theorem convergent}
  Denote $\bm{u}_{j}$ as the solution of (\ref{eq RCMRA4}) in the $(j-1)$-th iteration. Then, the sequence $\{\bm{u}_j\}$ satisfies the following properties:\\
  (1) $\mathcal{G}^{\epsilon}[T(\bm{u}_j)]$ is monotonically decreasing and bounded as $j\rightarrow +\infty$.\\
  (2) The sequence $\{\bm{u}_j\}$ converges to a local minimum of (\ref{eq RCMRA2}).
\end{theorem}
\begin{proof}
Please see Appendix \ref{appendix convergent}.
\end{proof}
Theorem \ref{theorem convergent} shows that by iteratively solving (\ref{eq RCMRA4}), we can finally obtain a local minimum of (\ref{eq RCMRA2}).

\section{Connection to Prior Arts}
\label{sec connection to SSR and atomic norm}

\subsection{Connection to the sparsity-based methods}
\label{subsection connection to SSR}

We start by extending the last equality of (\ref{eq RCMRA6}) to a convex one by using the sparse representation theory. Suppose that the whole angle space $[-90^{\circ},90^{\circ}]$ is divided by using a uniform grid of $N'$ points $\bm{\vartheta}' \triangleq \{\vartheta_1,\cdots,\vartheta_{N'}\}$ and further assume that the true DOAs $\bm{\theta}$ lie exactly on the grid, i.e., $\bm{\theta} \subset \bm{\vartheta}'$. Denote the corresponding manifold matrix and power by $\bm{A}'_{\bm{\Omega}} \in \mathbb{C}^{M\times N'}$ and $\bm{p}'=[p_1,\cdots,p_{N'}] \in \mathbb{R}^{N' \times 1}$, respectively, and $\tilde{\bm{A}}'_{\bm{\Omega}} = \bar{\bm{A}}'_{\bm{\Omega}} \odot \bm{A}'_{\bm{\Omega}}$ and $\bm{r}_{\bm{\Omega}} = \text{vec}(\hat{\bm{R}}_{\bm{\Omega}} - \sigma\bm{I})$. The resulting sparse model is shown as,
\begin{equation}\label{eq SSR model}
  \min_{\bm{p}' \succeq \bm{0}}\ \sum_n \frac{p_n}{\omega_n^{(j)}} \ \text{s.t.}\ \Big\|\bm{Q}\big( \bm{r}_{\bm{\Omega}}-\tilde{\bm{A}}'_{\bm{\Omega}} \bm{p}' \big)\Big\|_2^2 \leq \gamma,
\end{equation}
where $\omega_n^{(j)}$ denotes the $n$-th weight in the $j$-th iteration.
It is easy to see that model (\ref{eq SSR model}) is a reweighted $\ell_1$ norm minimization model where the weights $\{\omega_n\}$ are used to enhance the sparsity of the solution and improve the reconstruction performance \cite{reweightedl1norm2014TSP}. In particular, to promote sparsity, $\omega_n$ should be selected such that it has a large value when $p_n\neq 0$ and significantly smaller one elsewhere. With this strategy, after $\bm{p}'_{j}$ is determined in the $(j-1)$-th iteration, the weight $\omega_n^{(j)}$ can be obtained as,
\begin{equation}\label{eq weight update for SSR}
  \omega_n^{(j)} = \frac{1}{\bm{a}^H(\vartheta_n) \bm{W}_j^{\epsilon} \bm{a}(\vartheta_n)},
\end{equation}
where $\bm{W}_j^{\epsilon}$ can be computed by (\ref{eq RCMRA5}) and the estimated covariance $T(\bm{u}_j)$ of (\ref{eq RCMRA5}) in each iteration is obtained as $T(\bm{u}_j) = \sum_n p_{n}^{(j)} \bm{a}(\vartheta_{n})\bm{a}^H(\vartheta_{n})$. $\{\omega_n\}$ can be regarded as the power spectrum of the Capon's beamforming, which satisfies the selection condition of $\omega_n$ as shown in Section \ref{sec reweighted CMRA B}. By comparing the sparsity-based model (\ref{eq SSR model}) with (\ref{eq RCMRA6}), it can be easily concluded that model (\ref{eq RCMRA4}) is equivalent to (\ref{eq SSR model}) with $N' \rightarrow \infty$. In other words, the sparsity-based method (\ref{eq SSR model}) is a discretized version of model (\ref{eq RCMRA4}).
Finally, it is interesting to note that when $\bm{W}_j^{\epsilon} = \bm{I}$, i.e., $\omega_n = \frac{1}{N}$, model (\ref{eq SSR model}) is the sparsity-based model proposed in \cite{my2016TVT}.
It should be noted that since the weights $\{\omega_{n}\}$ enhance the sparsity, the reweighted $\ell_1$ norm iterative model (\ref{eq SSR model}) is expected to be superior to the $\ell_1$ norm model in \cite{my2016TVT}.

\subsection{Connection to Atomic Norm}

In this section, we attempt to interpret (\ref{eq RCMRA4}) as an ANM method. To do this, let us first define a weighted continuous dictionary,
\begin{equation}\label{eq weighted atomic dictionary}
  \mathcal{A}^{\omega} = \big\{\bm{B}^{\omega}(\theta)\big\} = \bigg\{ \frac{1}{N}\omega(\theta)\bm{a}(\theta)\bm{a}^H(\theta): \theta\in [-90^{\circ},90^{\circ}] \bigg\},
\end{equation}
where $\omega(\theta)\geq 0$ is a weighting function. Based on (\ref{eq weighted atomic dictionary}), the weighted atomic norm of $T(\bm{u})$ is defined as,
\begin{equation}\label{eq weighted atomic norm of Tu}
\begin{split}
  &\quad\|T(\bm{u})\|_{\mathcal{A}^{\omega}}\\
  &= \text{inf} \Big\{ \sum_k c_k^{\omega}: T(\bm{u})=\sum_k c_k^{\omega}\bm{B}^{\omega}(\theta_k),\bm{B}^{\omega}(\theta)\in \mathcal{A}^{\omega} \Big\}\\
  &= \text{inf} \Big\{ \sum_k c_k^{\omega}: T(\bm{u})=\sum_k c_k^{\omega}\omega(\theta_k)\bm{B}(\theta_k),\bm{B}(\theta)\in \mathcal{A} \Big\}\\
  &= \text{inf} \Big\{ \sum_kc_k\omega^{-1}(\theta_k): T(\bm{u})=\sum_kc_k\bm{B}(\theta_k),\bm{B}(\theta_k)\in \mathcal{A} \Big\}\\
  &= \sum_k N \omega^{-1}_k p_k\\
  &= N \text{tr}[\bm{W}^{\epsilon}T(\bm{u})],
\end{split}
\end{equation}
which indicates that model (\ref{eq RCMRA4}) is equivalent to the following ANM model,
\begin{equation}\label{eq ANM model of ICMRA}
\begin{split}
    &\min_{\bm{u}}\ \|T(\bm{u})\|_{\mathcal{A}^{\omega}} \\ &\quad\text{s.t.}\ \big\| \bm{Q} \text{vec}(\bm{E}_{\bm{\Omega}}) \big\|_2^2 \leq \beta^2,\ T(\bm{u}) \geq \bm{0}.
\end{split}
\end{equation}
According to the third equation in (\ref{eq weighted atomic norm of Tu}), an atom $\bm{B}(\theta)$, $\theta\in [-90^{\circ},90^{\circ}]$, is selected with a high probability if $\omega(\theta)$ is larger, which is the same conclusion as shown in Section \ref{subsection connection to SSR}.

\section{Computationally Efficient Implementations}
\label{sec computationally efficient implementations}

Here we present two implementation algorithms to speed up the convergence of the proposed method.

\subsection{Optimization via Duality}
\label{subsec duality}

We have empirically observed that a faster speed can be achieved when we solve the dual problem of model (\ref{eq RCMRA4}) as shown in the following.

First, by using the substitution $\bm{Y}= T(\bm{u})$, problem (\ref{eq RCMRA4}) can be reformulated as,
\begin{equation}\label{IV primal problem}
\begin{split}
  \min_{\bm{u},\bm{Y}} \ \text{tr}[\bm{W}_j^{\epsilon}\bm{Y})\big]\\
  \text{s.t.} \  \bm{Y} - T(\bm{u}) &= \bm{0},\\
  T(\bm{u}) &\geq \bm{0},\\
  \left\| \bm{Q} \text{vec} \big(\hat{\bm{R}}_{-\sigma} - \bm{Y}_{\bm{\Omega}}\big) \right\|_2^2 &\leq \beta^2,
\end{split}
\end{equation}
where $\hat{\bm{R}}_{-\sigma} = \hat{\bm{R}}_{\bm{\Omega}} - \sigma\bm{I}, \bm{Y}_{\bm{\Omega}} = \bm{\Gamma}\bm{Y}\bm{\Gamma}^T$.
Let $\bm{\Lambda}$, $\bm{V}$ and $\mu$ be the Lagrangian multipliers of the three constraints of (\ref{IV primal problem}), respectively. We can obtain the Lagrangian associated with the problem (\ref{IV primal problem}) as,
\begin{equation}\label{IV Lagrangian}
\begin{split}
  &\quad\ \mathcal{L}(\bm{u},\bm{Y},\bm{\Lambda},\bm{V},\mu)\\
  &= \text{tr}[\bm{W}_j^{\epsilon}\bm{Y}]-\text{tr}\left[\bm{\Lambda}\left(\bm{Y}-T(\bm{u})\right)\right] - \text{tr}\left[\bm{V}T(\bm{u})\right]\\
  &\quad + \mu \left\| \bm{Q}\text{vec}(\hat{\bm{R}}_{-\sigma} - \bm{Y}_{\bm{\Omega}}) \right\|^2_2 - \mu \beta^2\\
  &= \text{tr}\left[\left(\bm{W}_{j\bm{\Omega}}^{\epsilon}-\bm{\Lambda}_{\bm{\Omega}}\right)
  \bm{Y}_{\bm{\Omega}}\right] + \left(\bm{\omega}_{\overline{\bm{\Omega}}}-\bm{\lambda}_{\overline{\bm{\Omega}}}\right)^H \bm{y}_{\overline{\bm{\Omega}}}\\
  &\quad+ \text{tr}\left[\left(\bm{\Lambda}-\bm{V}\right)T(\bm{u})\right] + \mu \left\| \bm{Q}\text{vec}\left(\hat{\bm{R}}_{-\sigma} - \bm{Y}_{\bm{\Omega}}\right) \right\|^2_2 - \mu \beta^2,
\end{split}
\end{equation}
where $\bm{W}_{j\bm{\Omega}}^{\epsilon} = \bm{\Gamma} \bm{W}_j^{\epsilon} \bm{\Gamma}^T$ and $\bm{\Lambda}_{\bm{\Omega}} = \bm{\Gamma} \bm{\Lambda} \bm{\Gamma}^T$. Also, $\bm{\omega}_{\overline{\bm{\Omega}}}, \bm{\lambda}_{\overline{\bm{\Omega}}}$ and $\bm{y}_{\overline{\bm{\Omega}}}$ are the vectors composed of the entries of $\bm{W}_j^{\epsilon}, \bm{\Lambda} $ and $ \bm{Y}$, whose rows and columns are indexed by $\overline{\bm{\Omega}}$, respectively, in which $\overline{\bm{\Omega}} = \{ 1,\cdots,N \} - \bm{\Omega}$. Then, the Lagrange dual with respect to (\ref{IV primal problem}) can be easily formulated as follows,
\begin{equation}\label{IV Lagrange dual}
\begin{split}
  \mathcal{G} &= \min_{\bm{u},\bm{Y}} \max_{\bm{\Lambda},\bm{V},\mu} \mathcal{L}(\bm{u},\bm{Y},\bm{\Lambda},\bm{V},\mu)\\
  &= \max_{\bm{\Lambda},\bm{V},\mu} \min_{\bm{u},\bm{Y}} \mathcal{L}(\bm{u},\bm{Y},\bm{\Lambda},\bm{V},\mu)\\
  &= \max_{\bm{\Lambda},\bm{V},\mu} -\frac{1}{4\mu} \left\| \bm{Q}^{-H}\text{vec}(\bm{W}_{j\bm{\Omega}}^{\epsilon}-\bm{\Lambda}_{\bm{\Omega}}) \right\|^2_2 - \mu\beta^2\\
  &\quad\quad\quad\quad + \text{vec}\big(\hat{\bm{R}}_{-\sigma}\big)^H \text{vec}(\bm{W}_{j\bm{\Omega}}^{\epsilon}-\bm{\Lambda}_{\bm{\Omega}})\\
  &\quad\ \text{s.t.} \begin{cases}
        \bm{V} \geq \bm{0},\\
        T^*(\bm{\Lambda} - \bm{V}) = \bm{0},\\
        \bm{\lambda}_{\overline{\bm{\Omega}}} = \bm{\omega}_{\overline{\bm{\Omega}}},
        \end{cases}
\end{split}
\end{equation}
where $T^*(\bm{V}) = [v_{-(N-1)}, \cdots, v_{N-1}]^T$, with $v_n$ being the sum of the $n$-th diagonal of $\bm{V}\in \mathbb{C}^{N\times N}$. The second equality in (\ref{IV Lagrange dual}) holds because of strong duality \cite{boyd2004convex}. Then by noting that $\frac{1}{4\mu} \left\| \bm{Q}^{-H}\text{vec}(\bm{W}_{j\bm{\Omega}}^{\epsilon}-\bm{\Lambda}_{\bm{\Omega}}) \right\|^2_2 + \mu\beta^2 \geq \beta \left\| \bm{Q}^{-H}\text{vec}(\bm{W}_{j\bm{\Omega}}^{\epsilon}-\bm{\Lambda}_{\bm{\Omega}}) \right\|_2$, the dual problem of (\ref{IV primal problem}) can be formulated as,
\begin{equation}\label{IV dual problem}
    \begin{split}
        \min_{\bm{\Lambda},\bm{V}}\ &\beta \left\| \bm{Q}^{-H}\text{vec}(\bm{W}_{j\bm{\Omega}}^{\epsilon}-\bm{\Lambda}_{\bm{\Omega}}) \right\|_2\\
        &- \text{vec}\big(\hat{\bm{R}}_{-\sigma}\big)^H \text{vec}(\bm{W}_{j\bm{\Omega}}^{\epsilon}-\bm{\Lambda}_{\bm{\Omega}})\\
        &\text{s.t.} \begin{cases}
        \bm{V} \geq \bm{0},\\
        T^*(\bm{\Lambda} - \bm{V}) = \bm{0},\\
        \bm{\lambda}_{\overline{\bm{\Omega}}} = \bm{\omega}_{\overline{\bm{\Omega}}},
        \end{cases}
    \end{split}
\end{equation}
which is also convex and can be solved using CVX. Since the strong duality holds, the solution to problem (\ref{eq RCMRA4}) can be obtained as the dual variable of $\bm{V}$. As a result, the reweighted algorithm can be iteratively implemented and the DOAs can be estimated.

\subsection{A Fast Implementation for ICMRA}
\label{sec FICMRA}

Although solving the dual problem can save computations to some extent, the employed CVX solver is still time-consuming. In this section, we propose a computationally efficient method by deriving a closed-form solution.

In our model (\ref{eq RCMRA4}), the covariance matrix $T(\bm{u})$ is constrained to be positive semidefinite, which is almost sure with moderate or large number of snapshots. The research in \cite{AML1999LiHBTSP} indicates that $L\geq 15$ is large enough to ensure that the estimated covariance matrix is positive semidefinite in a $5$-element ULA case. Here we focus on moderate/large values of $L$, and thus drop the constraint $T(\bm{u})\geq \bm{0}$ in (\ref{eq RCMRA4}) for speed consideration and rewrite it into the Lagrangian form below,
\begin{equation}\label{eq Lagrangian of ICMRA}
  \min_{\bm{u}}\ \lambda\text{tr}\big[\bm{W}_j^{\epsilon}T(\bm{u})\big] + \frac{1}{2}\Big\| \bm{Q} \text{vec}\big(\hat{\bm{R}}_{-\sigma}-T_{\bm{\Omega}}(\bm{u})\big) \Big\|_2^2,
\end{equation}
where $\lambda>0$ is a lagrangian multiplier. We then rewrite (\ref{eq Lagrangian of ICMRA}) as,
\begin{equation}\label{eq1 fast ICMRA}
\begin{split}
  &\quad \min_{\bm{u}}\ \lambda\text{tr}\big[\bm{W}_j^{\epsilon}T(\bm{u})\big] + \frac{1}{2}\Big\| \bm{Q} \text{vec}\big(\hat{\bm{R}}_{-\sigma}-T_{\bm{\Omega}}(\bm{u})\big) \Big\|_2^2\\
  &= \min_{\bm{u}}\ \lambda\text{tr}\big[\bm{W}_j^{\epsilon}T(\bm{u})\big] + \frac{1}{2}\Big[ \text{vec}^H\big(\hat{\bm{R}}_{-\sigma} - T_{\bm{\Omega}}(\bm{u})\big)\\
  &\qquad\quad\times \text{vec}\Big(\hat{\bm{R}}_{-\sigma}^{-1} \big(\hat{\bm{R}}_{-\sigma} - T_{\bm{\Omega}}(\bm{u})\big) \hat{\bm{R}}_{-\sigma}^{-1}\Big) \Big]\\
  &= \min_{\bm{u}}\ \lambda\text{tr}\big[\bm{W}_j^{\epsilon}T(\bm{u})\big] + \frac{1}{2}\Big[ \text{tr}\big[ T_{\bm{\Omega}}(\bm{u}) \hat{\bm{R}}_{-\sigma}^{-1} T_{\bm{\Omega}}(\bm{u}) \hat{\bm{R}}_{-\sigma}^{-1}\big]\\
  &\qquad\quad- 2\text{tr} \big[ T_{\bm{\Omega}}(\bm{u}) \hat{\bm{R}}_{-\sigma} \big] \Big]\\
  &= \min_{\bm{u}}\ \text{tr}\big[ (\lambda \bm{W}_j^{\epsilon}-\bm{C}) T(\bm{u}) \big] + \frac{1}{2}\text{tr}\big[ T(\bm{u}) \bm{C} T(\bm{u}) \bm{C} \big],
\end{split}
\end{equation}
where $\bm{C} = \bm{\Gamma}^T \hat{\bm{R}}_{-\sigma}^{-1} \bm{\Gamma}$. By letting the derivative of the last objective function in (\ref{eq1 fast ICMRA}) with respect to $\bm{u}$ be zero, it can be shown that the optimal solution of (\ref{eq1 fast ICMRA}) satisfies the following equality,
\begin{equation}\label{eq2 fast ICMRA}
  T^*(\bm{C}-\lambda\bm{W}_j^{\epsilon}) = T^*(\bm{C}T(\bm{u})\bm{C}) .
\end{equation}
Clearly, (\ref{eq2 fast ICMRA}) is an $N$-variate linear equation which can be solved by the following procedure. First, the right-hand term of (\ref{eq2 fast ICMRA}) can be transformed as,
\begin{equation}\label{eq3 fast ICMRA}
  T^*(\bm{C}T(\bm{u})\bm{C}) = \underbrace{\left[\begin{matrix} \bar{\bm{\Phi}}_{-1} \\ \bm{\Phi}
  \end{matrix}\right]}_{\bm{Z}} \left[\begin{matrix} \bm{u}_{-1} \\ \bar{\bm{u}}
  \end{matrix}\right],
\end{equation}
where
\begin{equation}\label{eq phi_-1 fast ICMRA}
  \bm{\Phi}_{-1} = \big[\bm{\Phi}_{N,:}^T,\cdots,\bm{\Phi}_{2,:}^T\big]^T,
\end{equation}
\begin{equation}\label{eq u_-1 fast ICMRA}
  \bm{u}_{-1} = [u_N,\cdots,u_2]^T,
\end{equation}
and
\begin{equation}\label{eq Tstar solution in 1D case Phi}
\bm{\Phi} = \left[
  \begin{matrix}
    T^{*T}\Big(\bm{C}_{:,\{1,\cdots,N\}} \bm{C}_{\{1,\cdots,N\},:}\Big)\\
    T^{*T}\Big(\bm{C}_{:,\{1,\cdots,N-1\}} \bm{C}_{\{2,\cdots,N\},:}\Big)\\
    \vdots\\
    T^{*T}\Big(\bm{C}_{:,1} \bm{C}_{N,:}\Big)
  \end{matrix}
  \right],
\end{equation}
with $\bm{C}_{:,\mathcal{A}}$ and $\bm{C}_{\mathcal{B},:}$ denoting the columns and rows of matrix $\bm{C}$ indexed by sets $\mathcal{A}$ and $\mathcal{B}$, respectively.
Then, by letting
\begin{equation}\label{eq Z_1 fast ICMRA}
  \bm{Z}_1 = \textit{fl}(\bm{Z}_{:,\{1:N\}}),
\end{equation}
and
\begin{equation}\label{eq Z_2 fast ICMRA}
  \bm{Z}_2 = [\bm{0}, \bm{Z}_{:,\{N+1:2N-1\}}],
\end{equation}
where the operator $\textit{fl}(\bm{Z})$ returns $\bm{Z}$ with row preserved and columns flipped in the left/right direction, we can rewrite (\ref{eq3 fast ICMRA}) into a more compact form,
\begin{equation}\label{eq Tstar solution in 1D case Z}
  T^*\Big(\bm{C} T(\bm{u}) \bm{C}\Big) = [\bm{Z}_1\ \bm{Z}_2] \left[
  \begin{matrix}
    \bm{u}\\
    \bar{\bm{u}}
  \end{matrix}
  \right].
\end{equation}
To obtain the estimate $\bm{u}$ which is complex-valued, we first transform (\ref{eq Tstar solution in 1D case Z}) into a real-valued matrix form as follows,
\begin{equation}\label{eq final equality in 1D case}
  \underbrace{\begin{bmatrix}
  \text{Re}(\bm{h})\\
  \text{Im}(\bm{h})
  \end{bmatrix}}_{\bm{h}_r} = \underbrace{\begin{bmatrix}
  \text{Re}(\bm{Z}_1+\bm{Z}_2) & \text{Im}(\bm{Z}_2-\bm{Z}_1)\\
  \text{Im}(\bm{Z}_1+\bm{Z}_2) & \text{Re}(\bm{Z}_1-\bm{Z}_2)
  \end{bmatrix}}_{\bm{Z}_r} \underbrace{\begin{bmatrix}
  \text{Re}(\bm{u})\\
  \text{Im}(\bm{u})
  \end{bmatrix}}_{\bm{u}_r},
\end{equation}
where $\bm{h}$ denotes the left-hand term of (\ref{eq2 fast ICMRA}) and the subscript $r$ denotes that the variable is real-valued. It is seen from (\ref{eq final equality in 1D case}) that $\bm{u}$ can be easily obtained from $\bm{u}_r = \bm{Z}_r^{\dag}\bm{h}_r$, where $\dag$ is the pseudo-inverse operator. Compared to using CVX, the derived closed-form solution can reduce the computational complexity to a great extent and hence the method is termed as fast ICMRA (FICMRA). Its performance and superiorities over other methods will be shown in the following section.

\section{Simulation Results}
\label{sec simulation}

In this section, we evaluate the performance of (F)ICMRA with comparison to CMRA \cite{my2016TVT}, MUSIC \cite{MUSIC1986schmidt}, $\ell_1$ singular value decomposition (L1-SVD) \cite{malioutov2005sparse} and sparse and parametric approach (SPA) \cite{yangzai2014SPA}. In our simulations, ICMRAs are implemented by solving the dual problem as discussed in Section \ref{subsec duality} and FICMRAs are carried out by the closed-form solution as derived in Section \ref{sec FICMRA}. We especially consider the closely adjacent signal cases which require high resolution. The nonconvex penalties employed in ICMRA are Logarithm, $\ell_p$ norm and Laplace and the corresponding proposed methods are $\text{(F)ICMRA}_{\text{log}}$, $\text{(F)ICMRA}_{\ell_p}$ and $\text{(F)ICMRA}_{\text{lap}}$.
For the initialization of ICMRA, we set $\bm{u}_0 = \bm{0}$ and $\epsilon_0 = 1$, hence the first iteration of ICMRA is equivalent to CMRA.
It is expected that other initializations may lead to different estimate $\hat{\bm{u}}$. From this point of view, we carry out an empirical analysis of initialization impact on the solution in Section \ref{sec initialization}, and show that starting with a zero vector is an appropriate attempt.
For the proposed methods, at each iteration, unless otherwise stated, the value of $\epsilon$ is reduced by $\epsilon_{j+1} = \frac{\epsilon_j}{\delta}$ with $\delta=2,10,10$ for Logarithm, $\ell_p$ and Laplace, respectively (the choice of $\delta$ will be discussed in Section \ref{sec simulation A}). The iteration stops if the maximum number of iterations, set to $20$, is reached, or the relative change of $\hat{\bm{u}}$ at two consecutive iterations is less than $10^{-4}$, i.e., $\frac{ \| \hat{\bm{u}}_{j+1}-\hat{\bm{u}}_j\|_F }{\|\hat{\bm{u}}_j\|_F} < 10^{-4}$. The number of sources is assumed to be unknown for the compared methods except MUSIC. The searching step of MUSIC is determined by the SNR, i.e., $\big( 10^{-\frac{\text{SNR}}{20}-1} \big)^{\circ}$. For the initialization of the three FICMRAs, $\lambda$ is set to $ 0.1$ which gives good performance empirically. Other settings of FICMRAs are the same as those of the corresponding ICMRAs.

\subsection{An illustration Example}
\label{sec simulation A}

\begin{figure}[!t]
\centering
\includegraphics[width=3in]{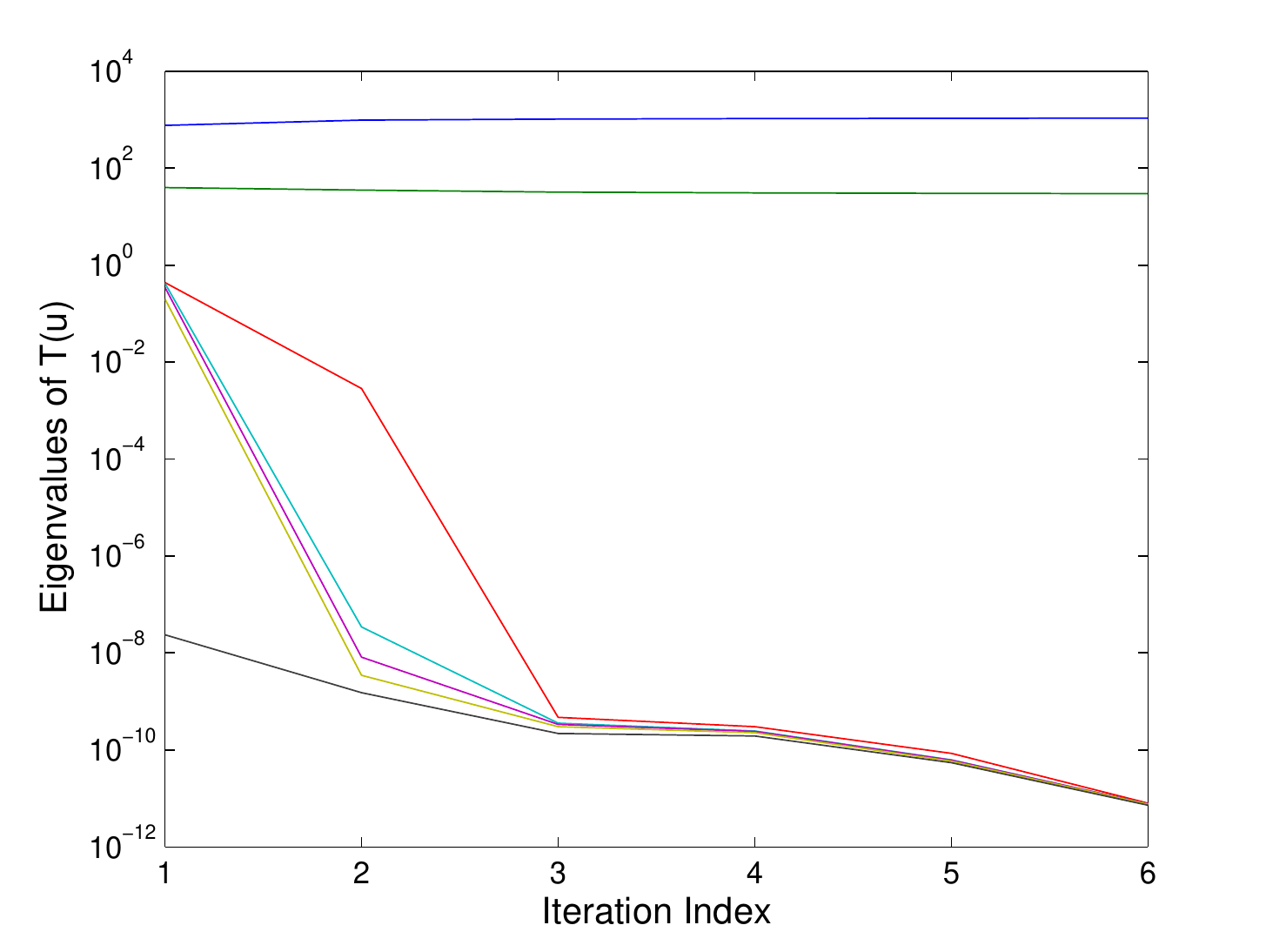}
\caption{The variation of eigenvalues of $T(\bm{u})$ with respect to the iteration index. The number of snapshots is set to $200$ and SNR$=20$dB.}
\label{fig eigenvalues}
\end{figure}

\begin{figure}[!t]
\begin{minipage}[b]{.48\linewidth}
  \centering
  \centerline{\includegraphics[width=1.9in]{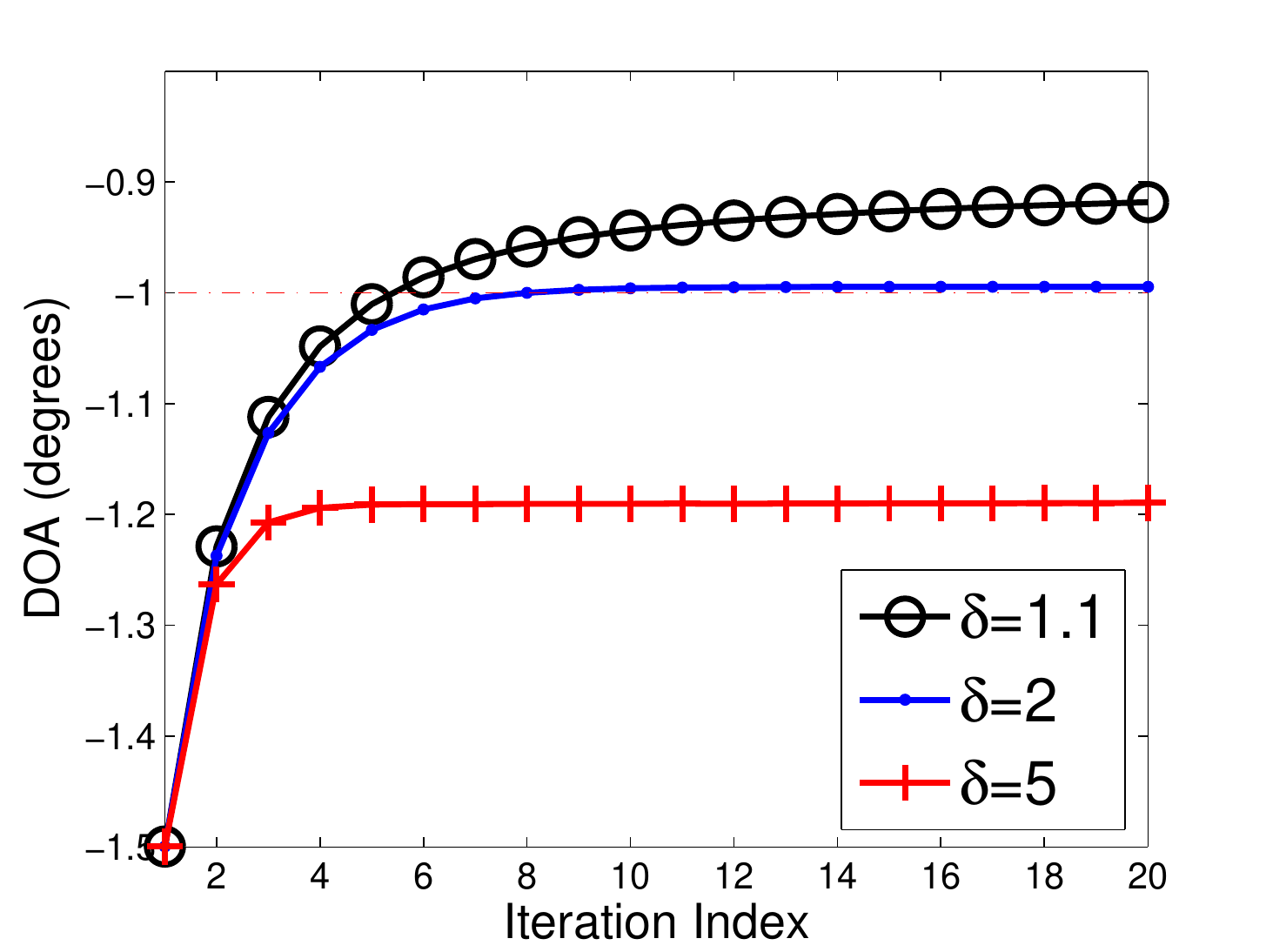}}
  \centerline{(a) Signal $1$}\medskip
\end{minipage}
\hfill
\begin{minipage}[b]{0.48\linewidth}
  \centering
  \centerline{\includegraphics[width=1.9in]{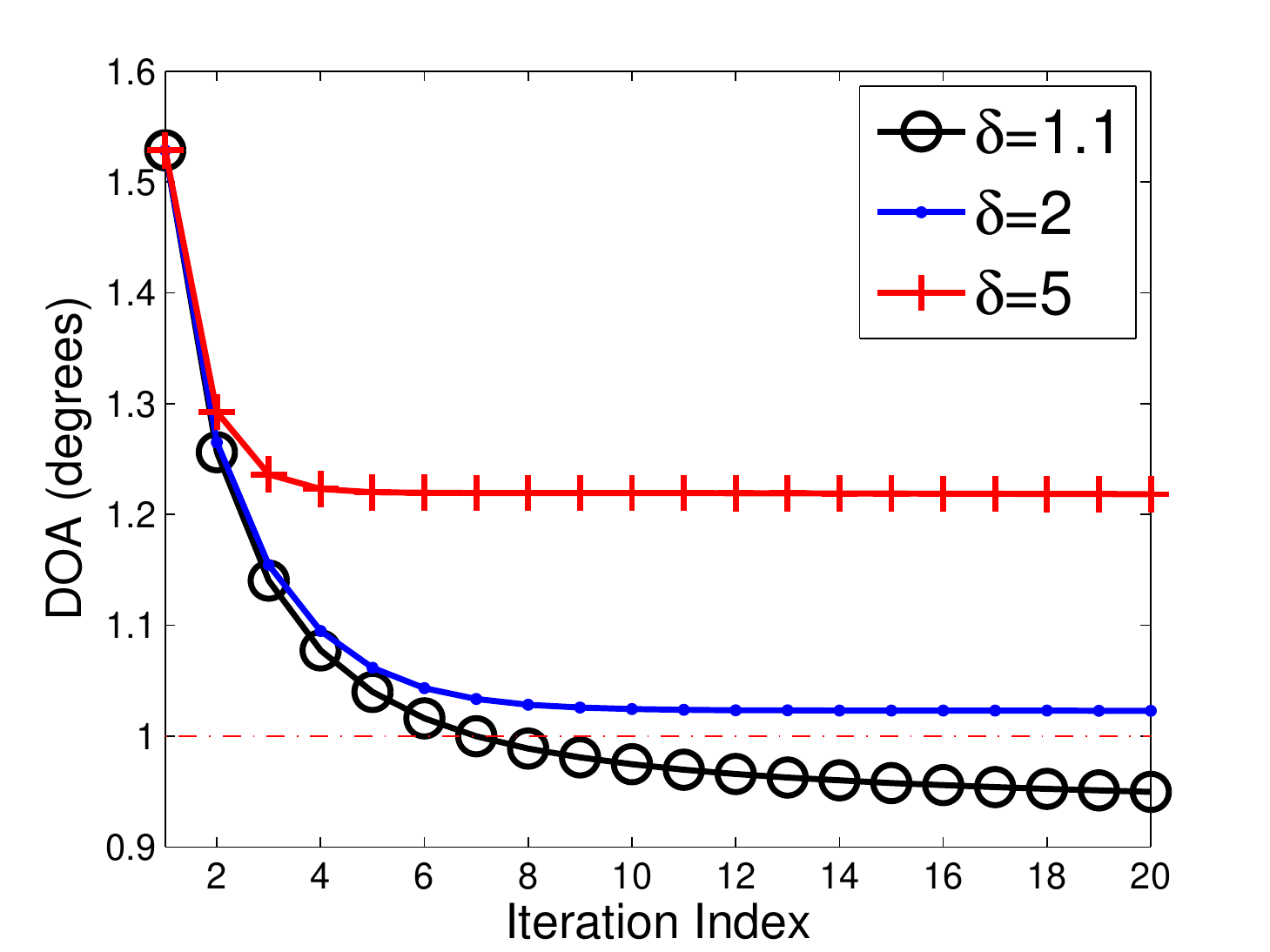}}
  \centerline{(b) Signal $2$}\medskip
\end{minipage}
\caption{Iteration results of $\text{ICMRA}_{\text{log}}$.}
\label{fig iteration}
\end{figure}

We first carry out a simple example to illustrate the iterative process and choose $\text{ICMRA}_{\text{log}}$ as the representative method. Assume two narrowband far-field signals impinge onto a $7$-element ULA from directions of $[-1^{\circ}, 1^{\circ}]$. Two hundred snapshots are collected for DOA estimation and the $\text{SNR}$ is set to $20$dB. To illustrate that our iterative procedure is able to promote the sparsity structure, we record the variation of eigenvalues of $T(\bm{u}_{j})$ with respect to the iteration index and plot them in Fig. \ref{fig eigenvalues}, in which different color curves denote different eigenvalues. Note that ICMRA is terminated after $10$ iterations and we only plot the eigenvalues of the first $6$ iterations for better visual effects. From Fig. \ref{fig eigenvalues} it can be seen that after the first iteration (a.k.a. the CMRA method), there exist two large eigenvalues, four moderate eigenvalues and only one extremely small eigenvalue. In other words, the estimated covariance matrix is singular but has at least rank-six. Subsequently, the moderate eigenvalues gradually decrease and approach zero (within numerical precision) after the third iteration while the other two large eigenvalues nearly remain unchanged. Hence by noting that the final covariance matrix estimated by ICMRA has a rank of two, which is equivalent to the number of signals, it can be concluded that the sparse structure is promoted. The ICMRA with two other penalties show similar performance and the details are omitted here.

\begin{figure*}[!t]
\centerline{
\subfloat[CMRA]{\includegraphics[width=1.8in]{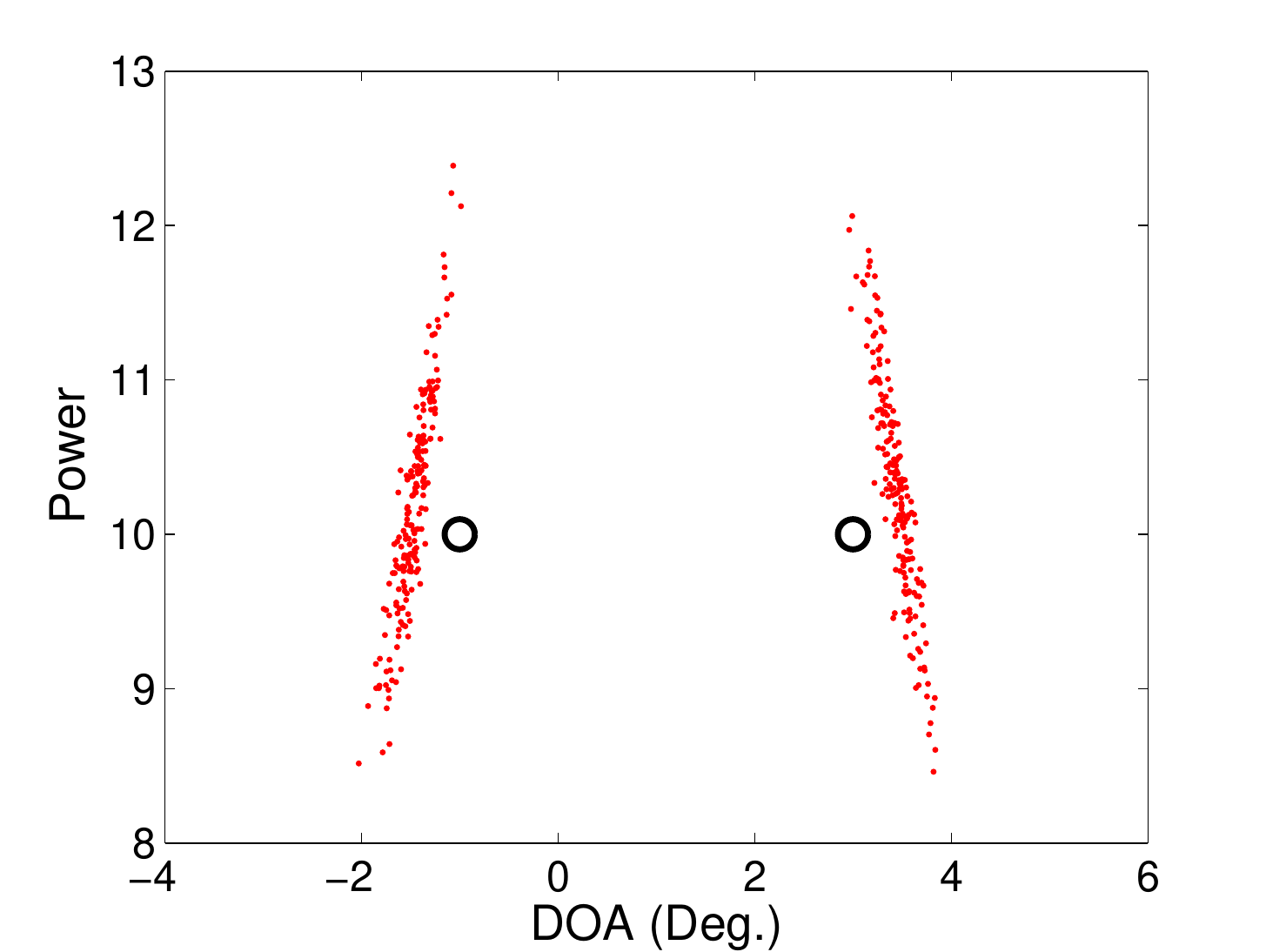}}
\hfil
\subfloat[$\text{ICMRA}_{\text{log}}$]{\includegraphics[width=1.8in]{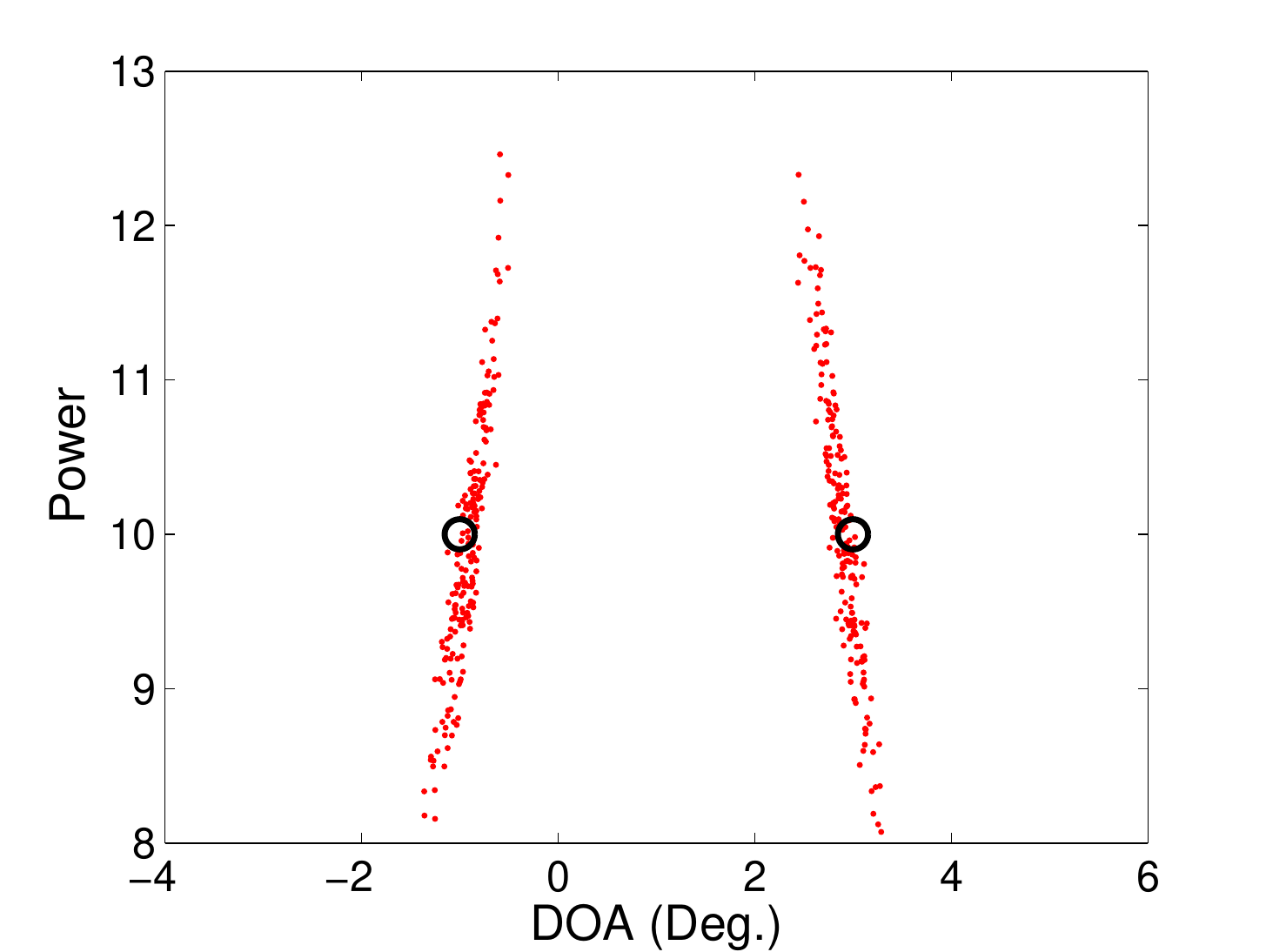}}
\hfil
\subfloat[$\text{ICMRA}_{\text{lp}}$]{\includegraphics[width=1.8in]{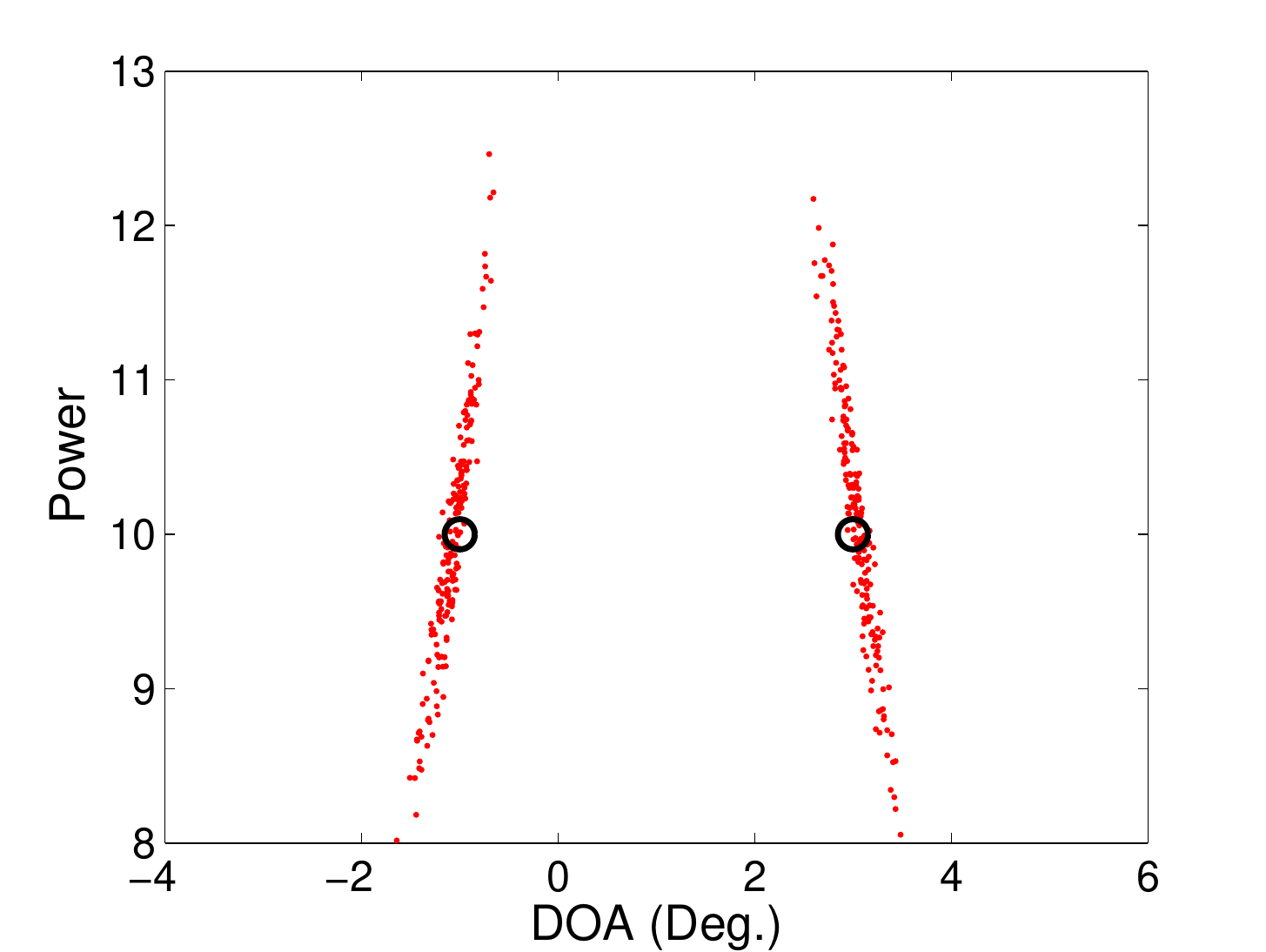}}
\hfil
\subfloat[$\text{ICMRA}_{\text{lap}}$]{\includegraphics[width=1.8in]{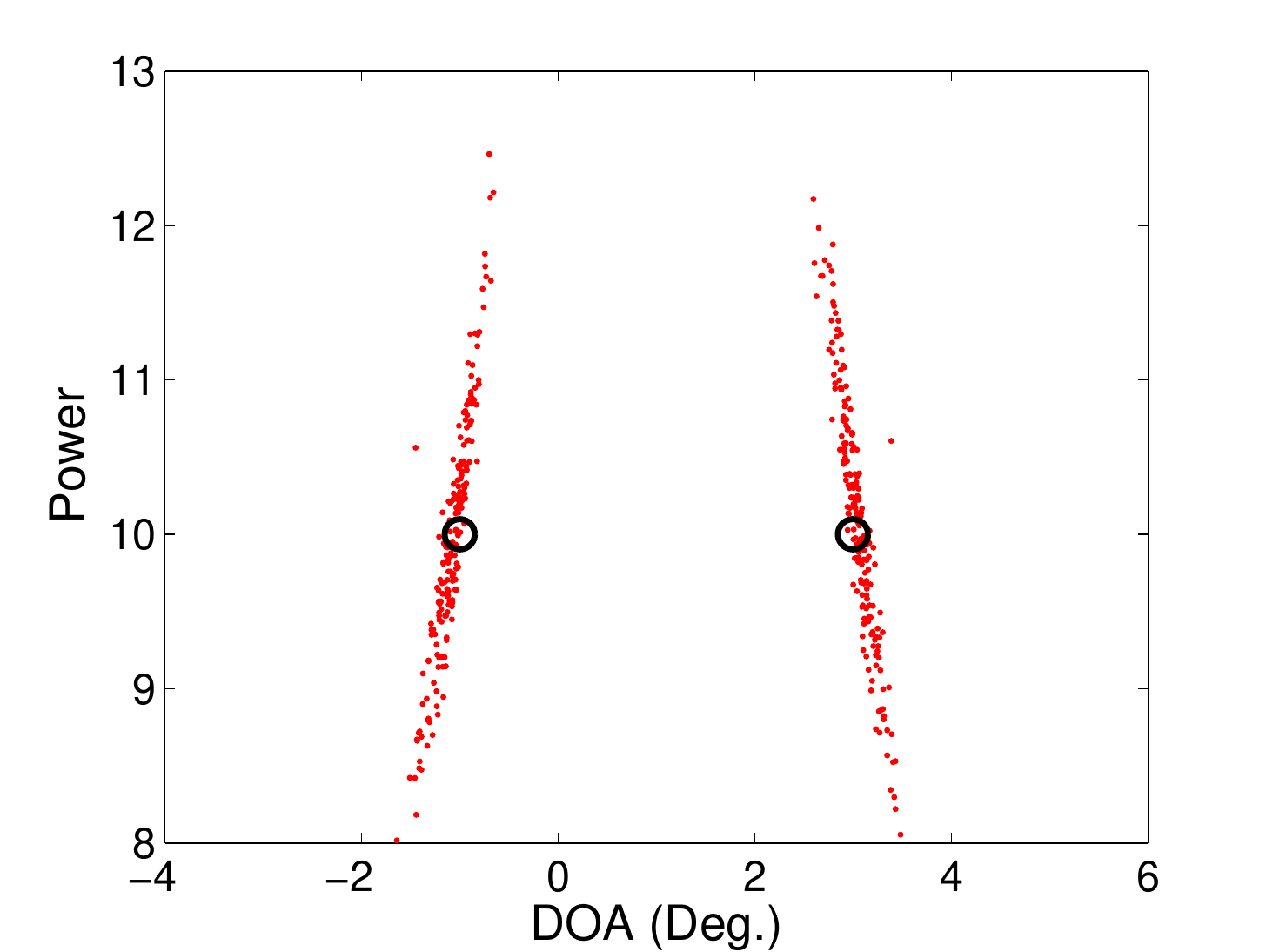}}}
\centerline{
\subfloat[$\text{FICMRA}_{\text{log}}$]{\includegraphics[width=1.8in]{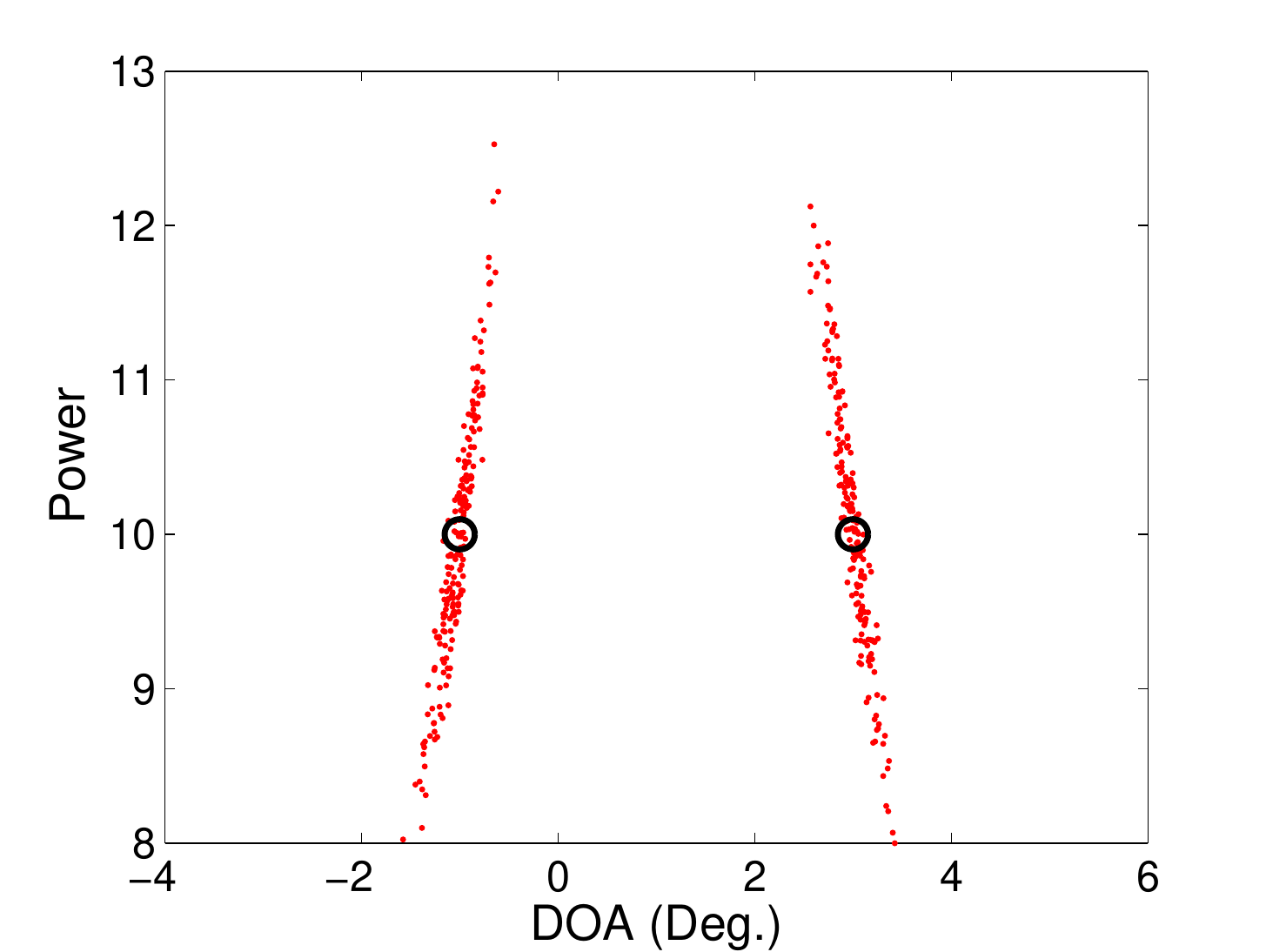}}
\hfil
\subfloat[$\text{FICMRA}_{\text{lp}}$]{\includegraphics[width=1.8in]{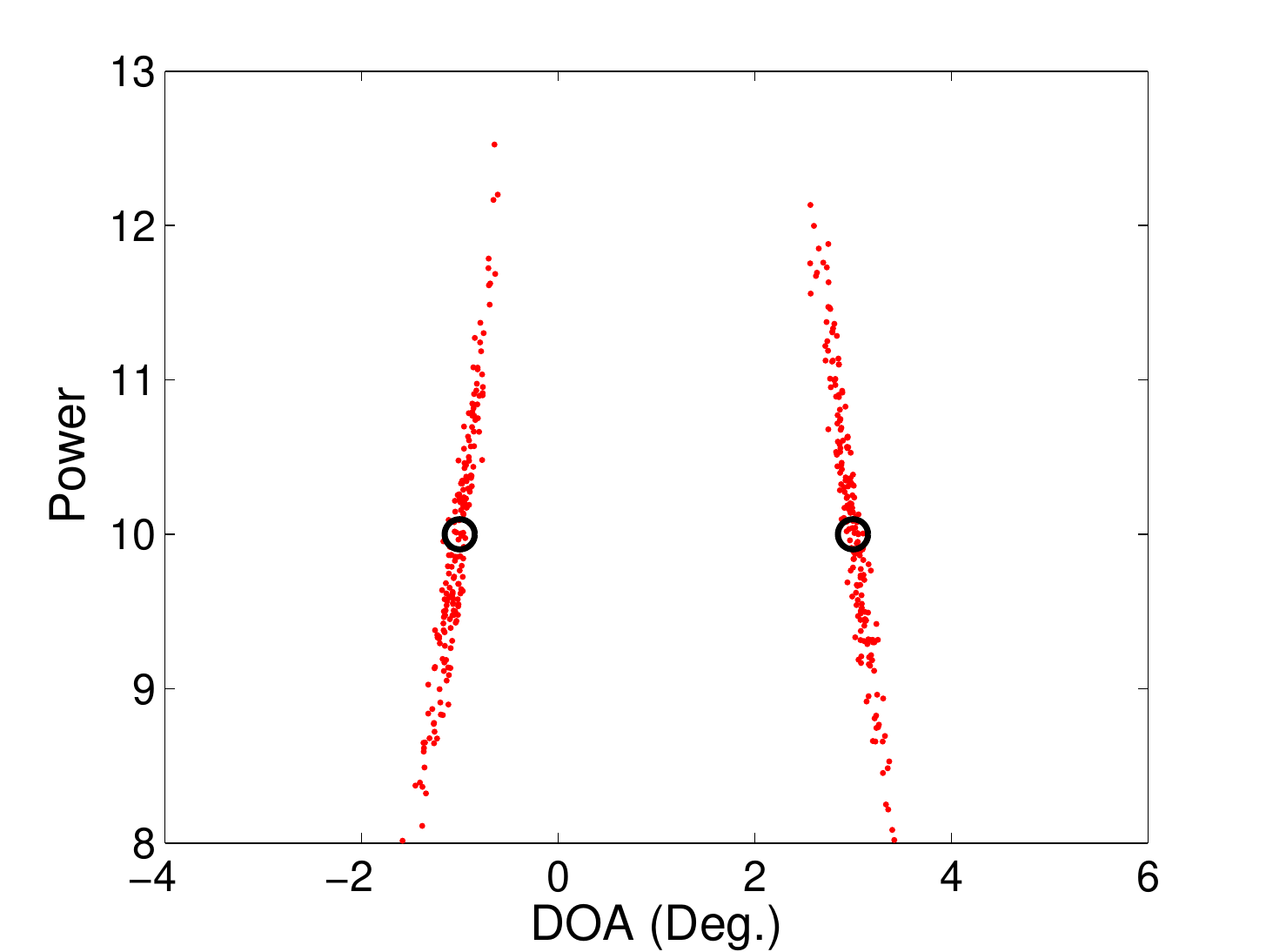}}
\hfil
\subfloat[$\text{FICMRA}_{\text{lap}}$]{\includegraphics[width=1.8in]{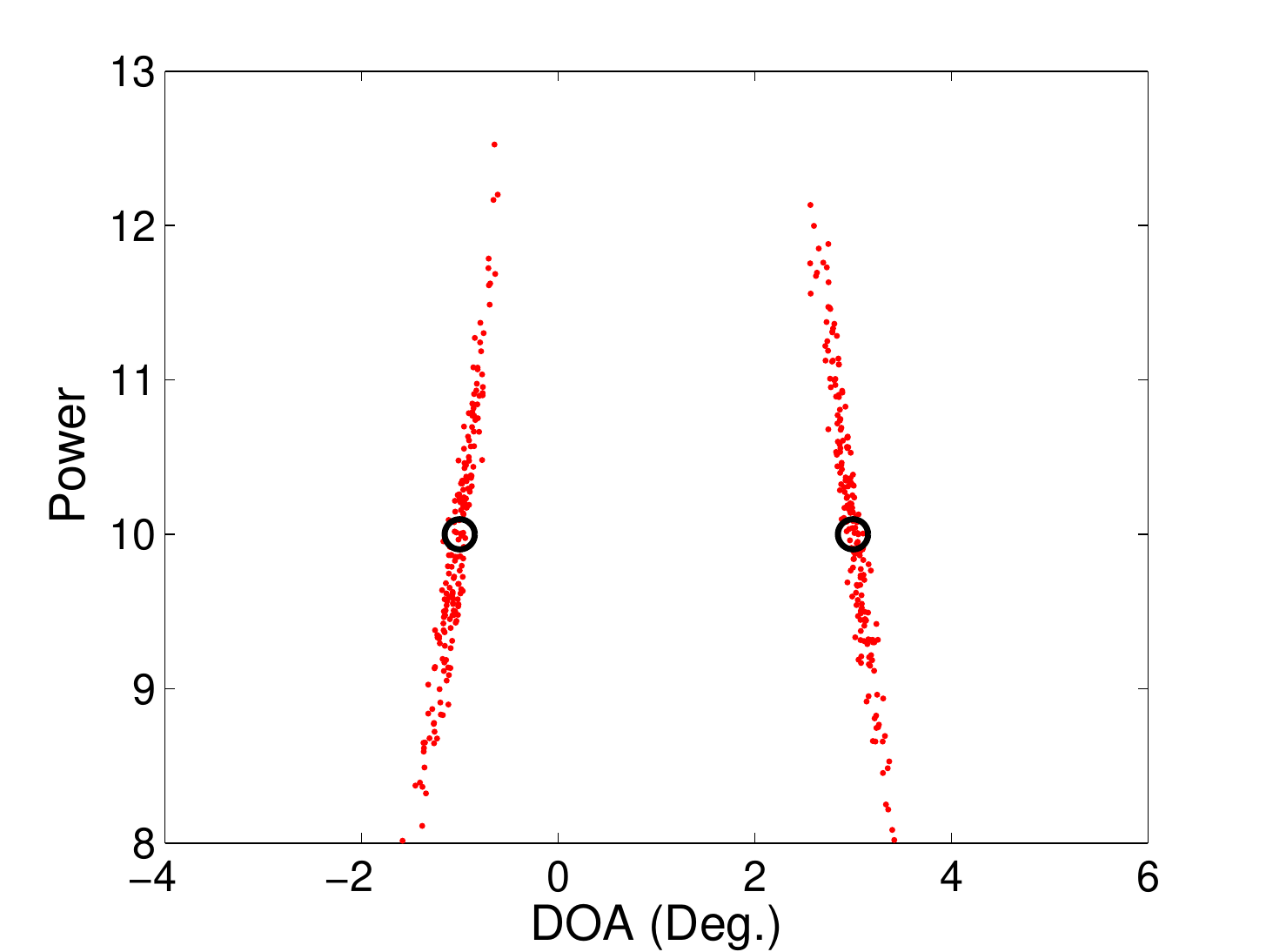}}}
\caption{DOA estimation comparison of CMRA, ICMRA and FICMRA for two uncorrelated sources impinging from $[-1^{\circ}, 3^{\circ}]$ with $L=400$ and SNR$=10$dB.}
\label{fig DOA estimation comparison}
\end{figure*}

\begin{figure*}[!t]
\centerline{
\subfloat[CMRA]{\includegraphics[width=1.8in]{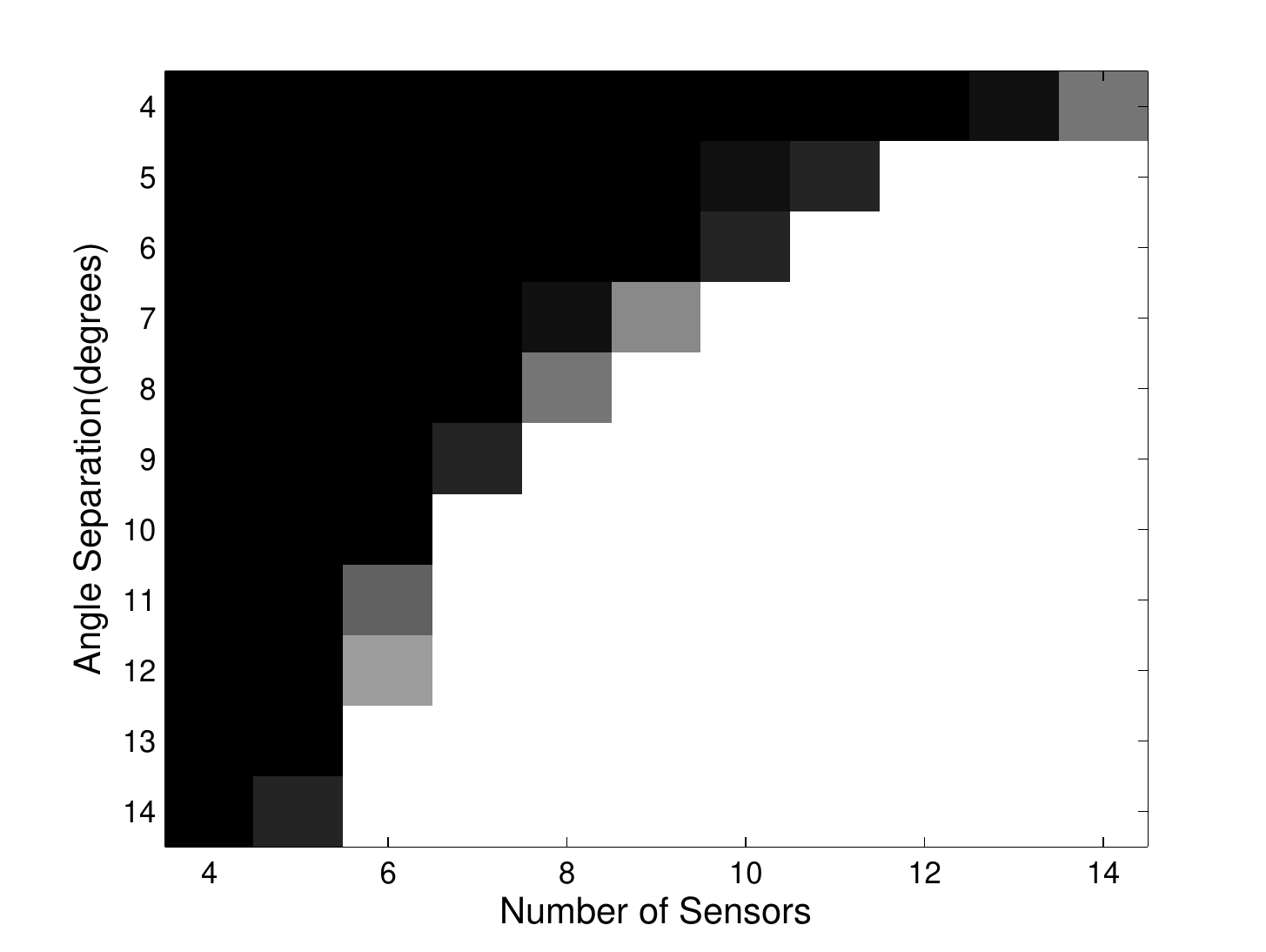}}
\hfil
\subfloat[$\text{ICMRA}_{\text{log}}$]{\includegraphics[width=1.8in]{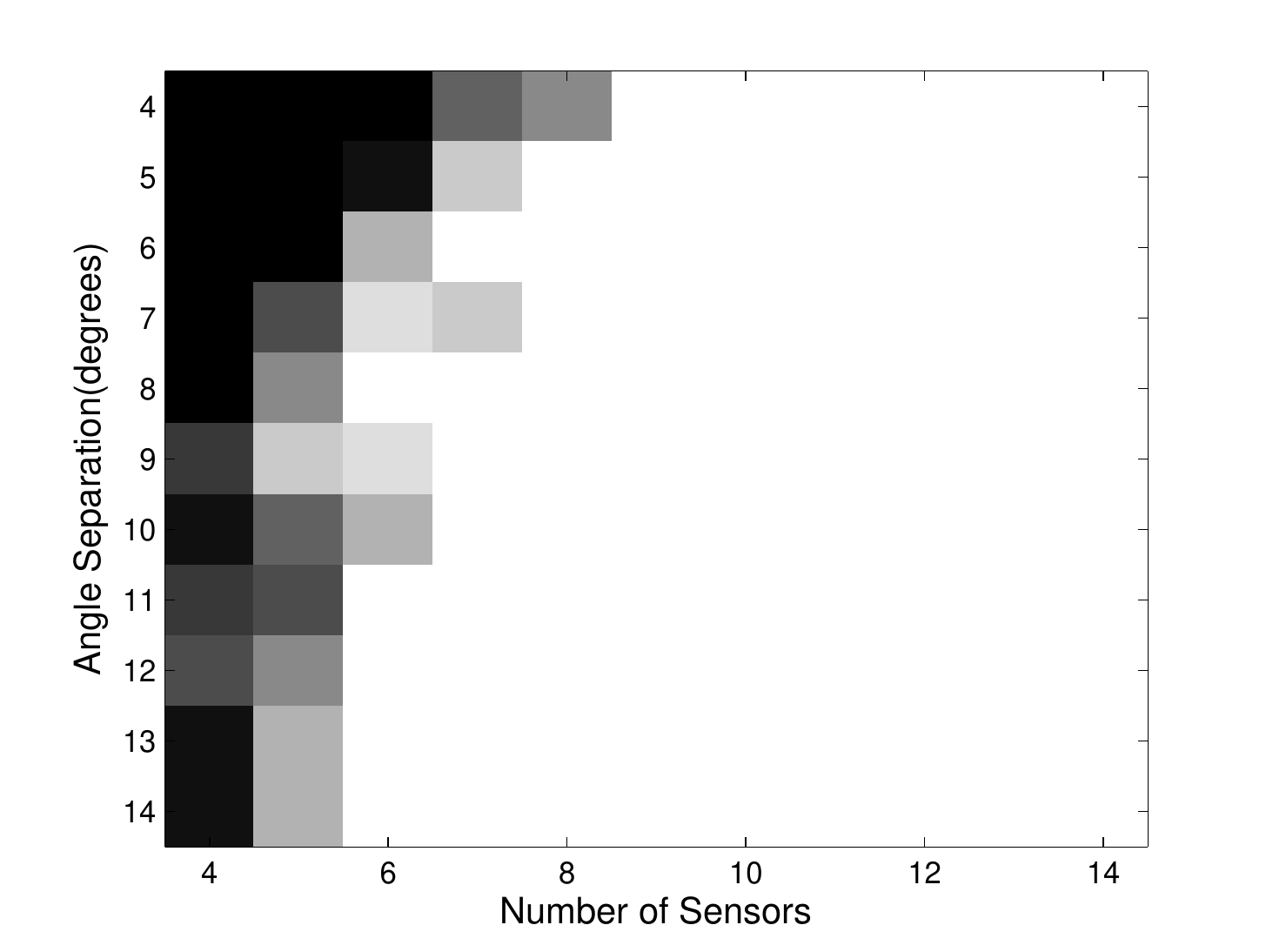}}
\hfil
\subfloat[$\text{ICMRA}_{\ell_p}$]{\includegraphics[width=1.8in]{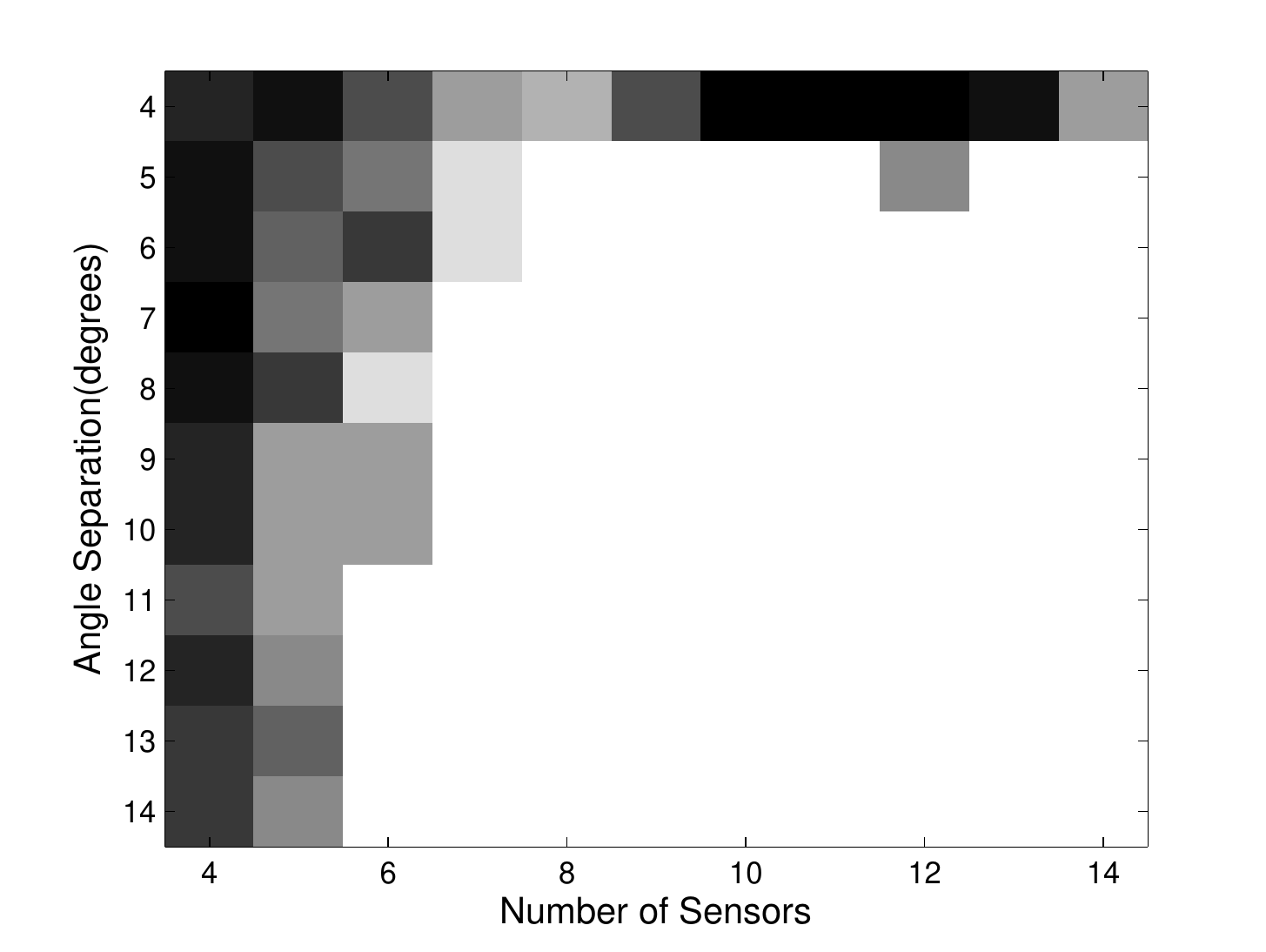}}
\hfil
\subfloat[$\text{ICMRA}_{\text{lap}}$]{\includegraphics[width=1.8in]{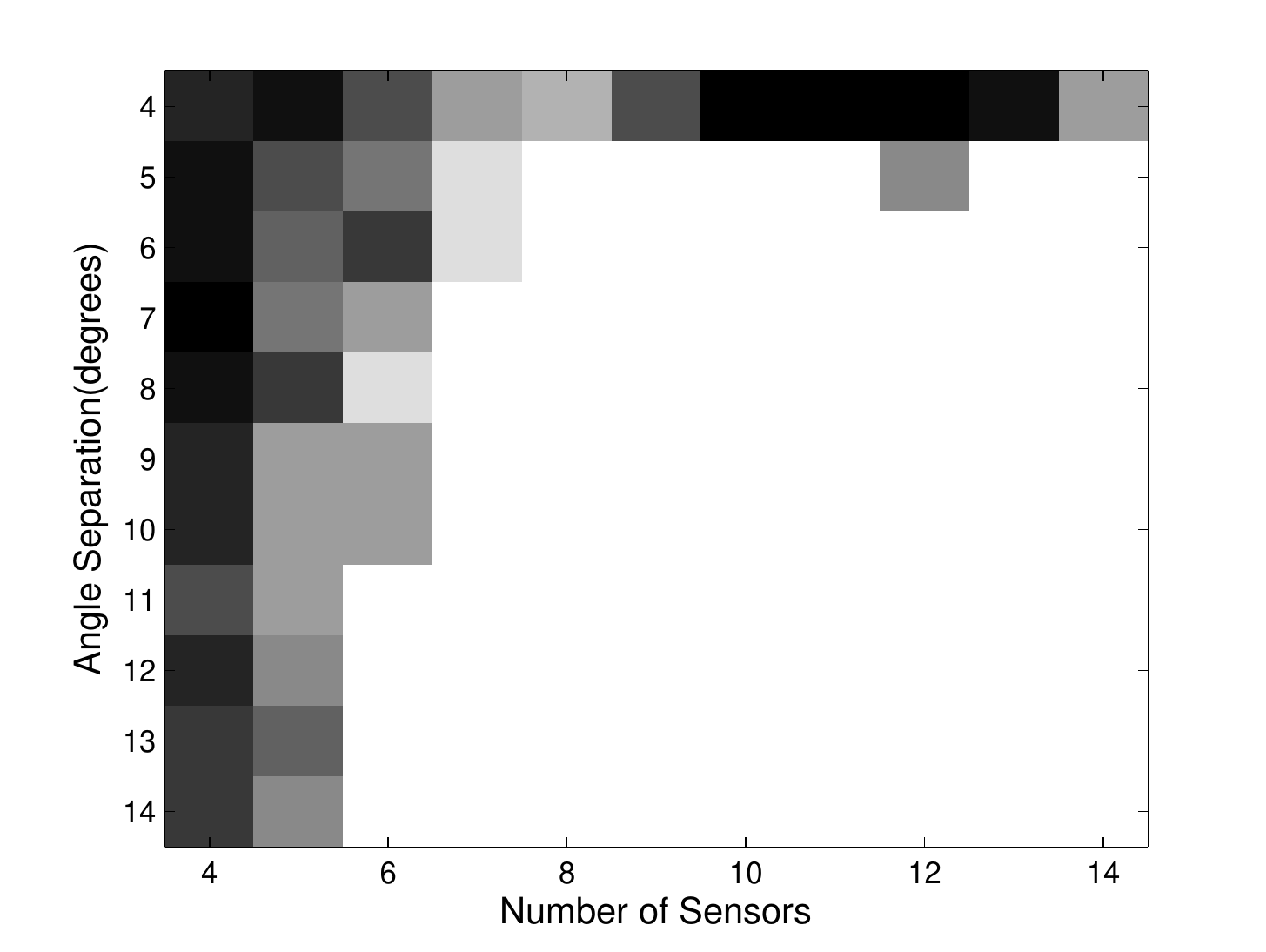}}}
\centerline{
\subfloat[$\text{FICMRA}_{\text{log}}$]{\includegraphics[width=1.8in]{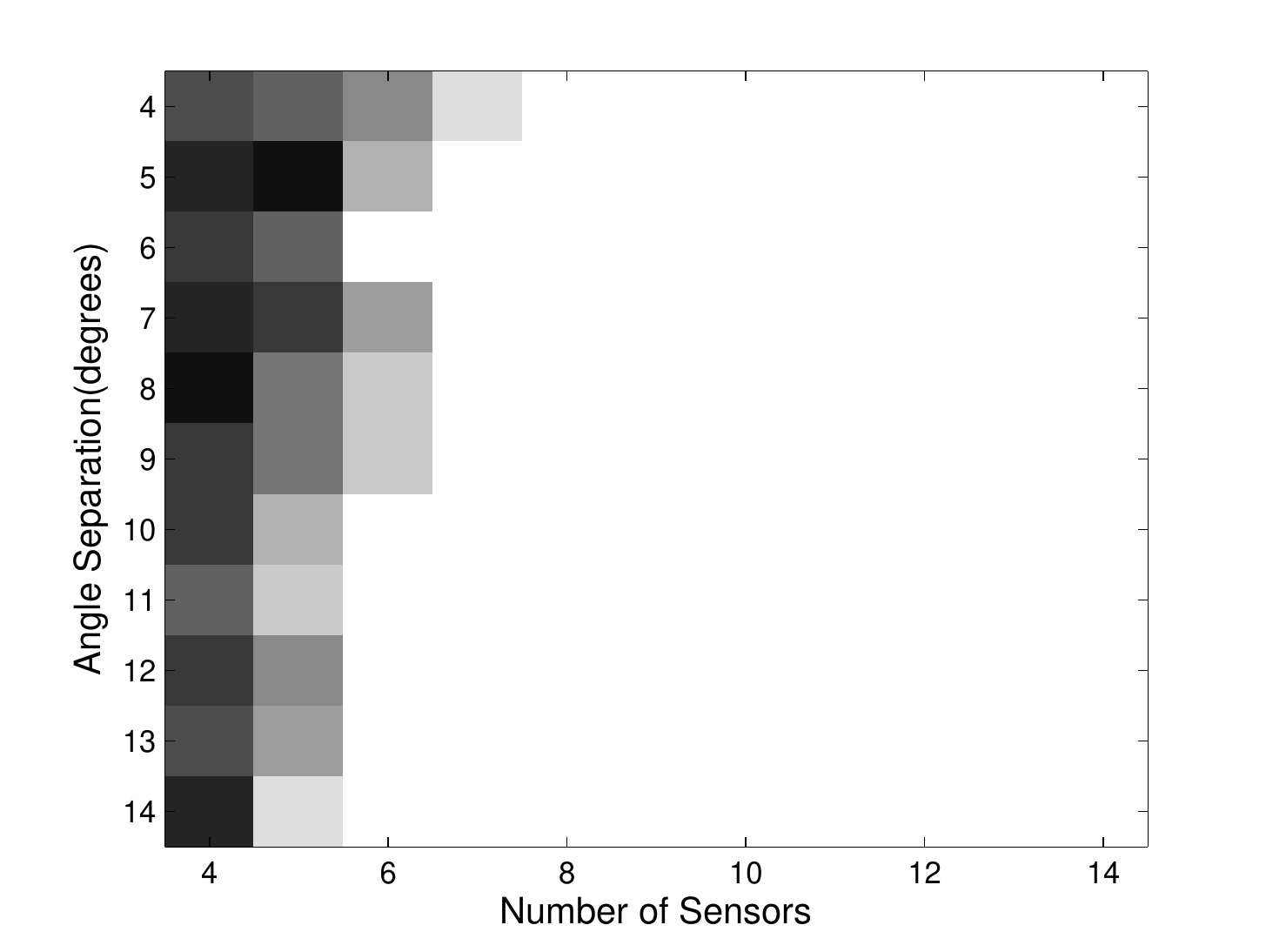}}
\hfil
\subfloat[$\text{FICMRA}_{\text{lp}}$]{\includegraphics[width=1.8in]{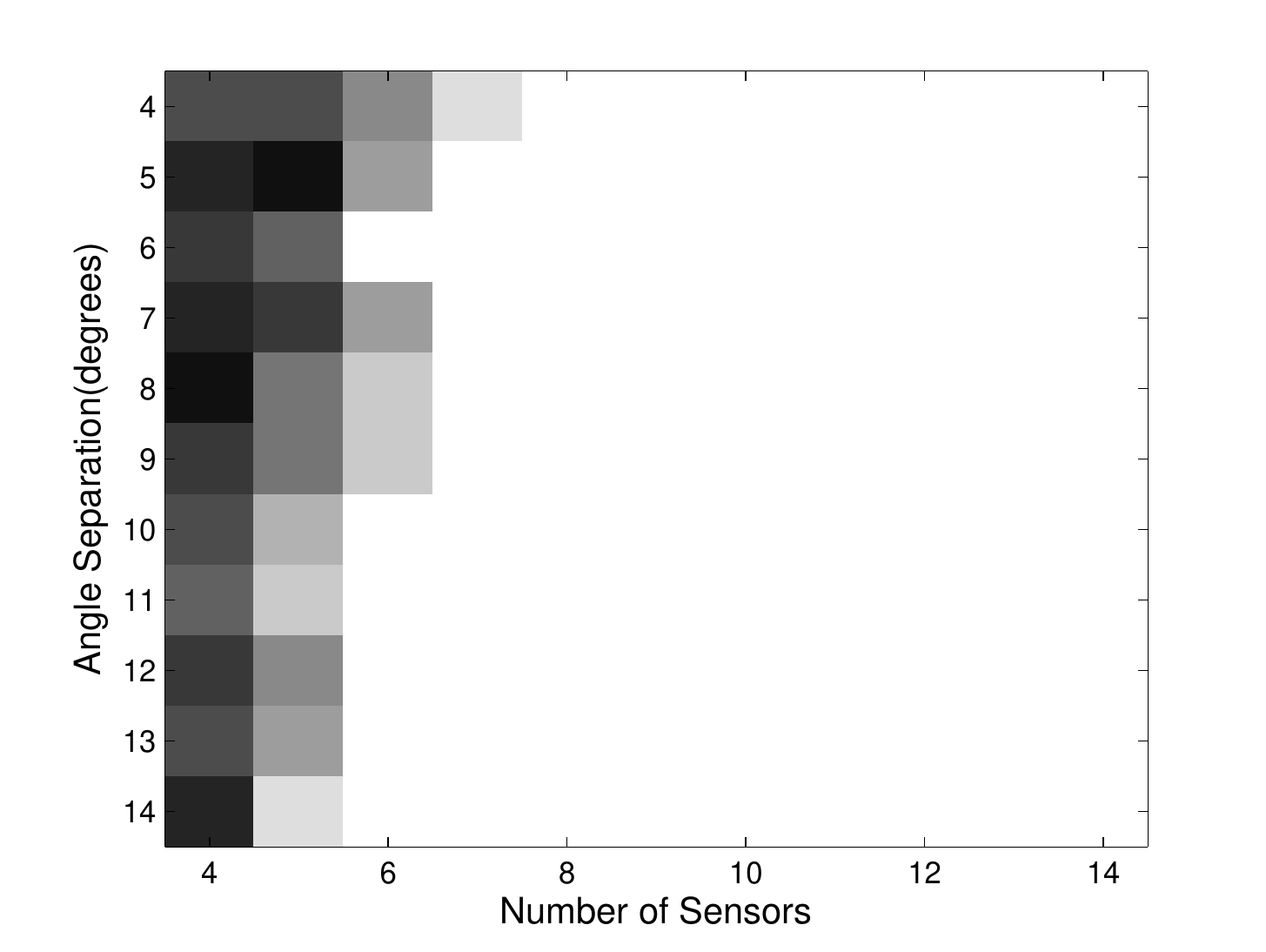}}
\hfil
\subfloat[$\text{FICMRA}_{\text{lap}}$]{\includegraphics[width=1.8in]{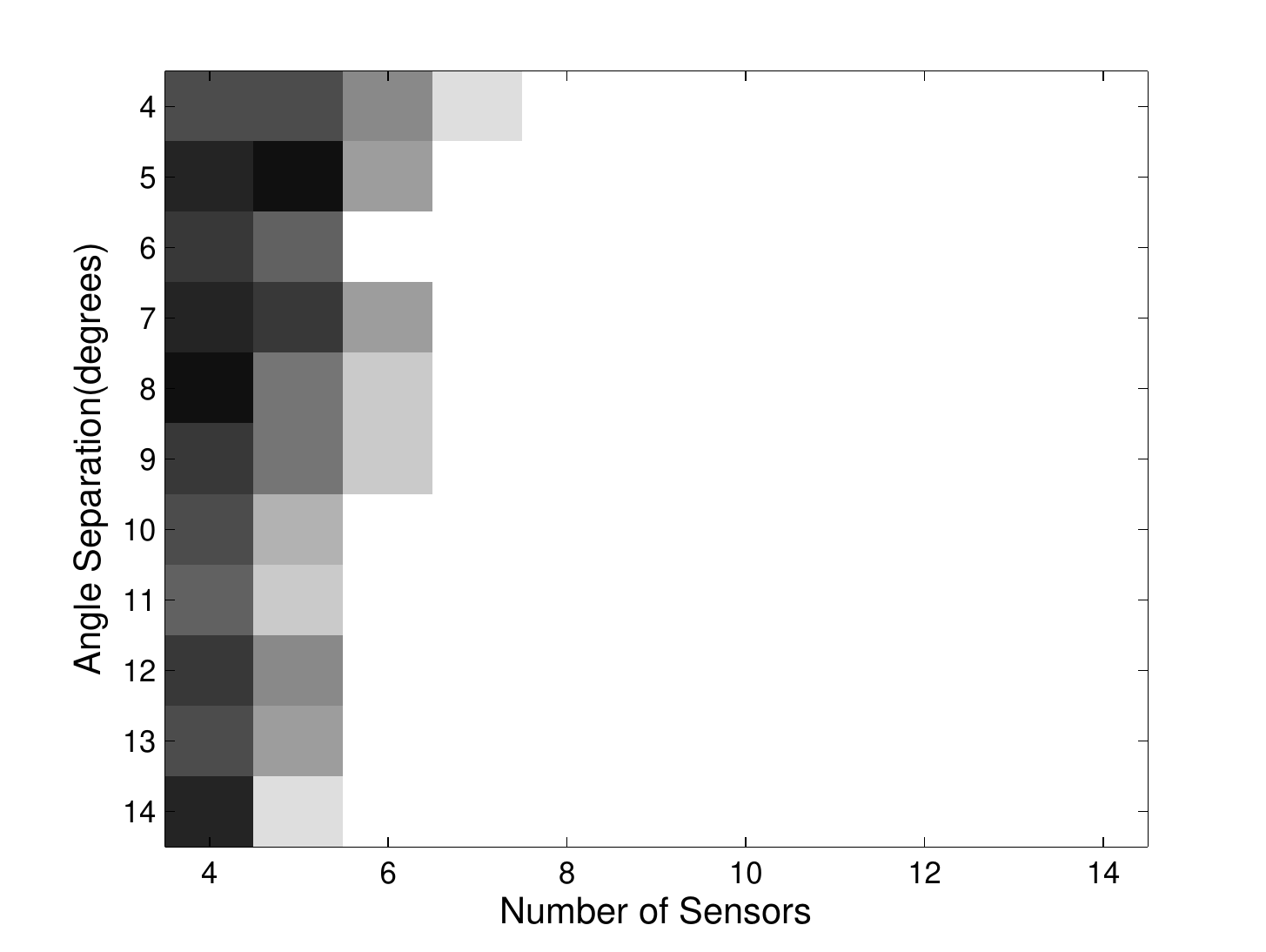}}}
\caption{Success rates of CMRA ICMRA and FICMRA with $L=400$, SNR$=10$dB.}
\label{fig grayscale}
\end{figure*}

\begin{table*}[!t]
\renewcommand{\arraystretch}{1.5}
\caption{CPU Time (in seconds) comparison of CMRA and the proposed methods. }
\label{table CPUTime}
\centering
\begin{tabular}{cccccccc}
\toprule
& CMRA & $\text{ICMRA}_{\text{log}}$ & $\text{ICMRA}_{\ell_p}$ & $\text{ICMRA}_{\text{lap}}$ & $\text{FICMRA}_{\text{log}}$ & $\text{FICMRA}_{\ell_p}$ & $\text{FICMRA}_{\text{lap}}$ \\
\midrule
CPU Time & 0.1670 & 5.4725 & 4.6653 & 2.1171 & 0.0071 & 0.0071 & 0.0033 \\
\bottomrule
\end{tabular}
\end{table*}

Next, we discuss the choice of $\delta$ and evaluate the influence of different values of $\delta$ on estimation performance. In particular, we repeat the previous simulation with $\delta=1.1,2,5$ and show the DOA estimates with respect to each iteration in Fig. \ref{fig iteration}. The iterations start from CMRA and are able to provide more accurate estimates.
Moreover, it can be observed that $\delta=5$ may lead to a faster convergence speed but subject to a worse estimate (more experiments indicate a larger $\delta$, e.g., $\delta=10$, may result in some outliers). Among the three choices of $\delta$, it is interesting to note that $\delta=1.1$ has the slowest convergence speed whereas does not produce the best estimates. In contrast, $\delta=2$ shows the best estimation performance compared to other choices of $\delta$. Hence $\delta=2$ is a proper choice for $\text{ICMRA}_{\text{log}}$ in terms of both accuracy and speed convergence. The choices of $\delta$ for other proposed methods as shown above have also been empirically and carefully determined. In computational speed, for each trial, CMRA (a.k.a. the first iteration of $\text{ICMRA}_{\text{log}}$) takes $0.25$s on average while $\text{ICMRA}_{\text{log}}$ requires $5.97$s because of the iterative procedure. In conclusion, compared to the CMRA, the reweighted iteration strategy is able to improve the resolution at the cost of more expensive computations.


\subsection{DOA Estimation Performance}

\begin{figure}[!t]
\centering
\includegraphics[width=3in]{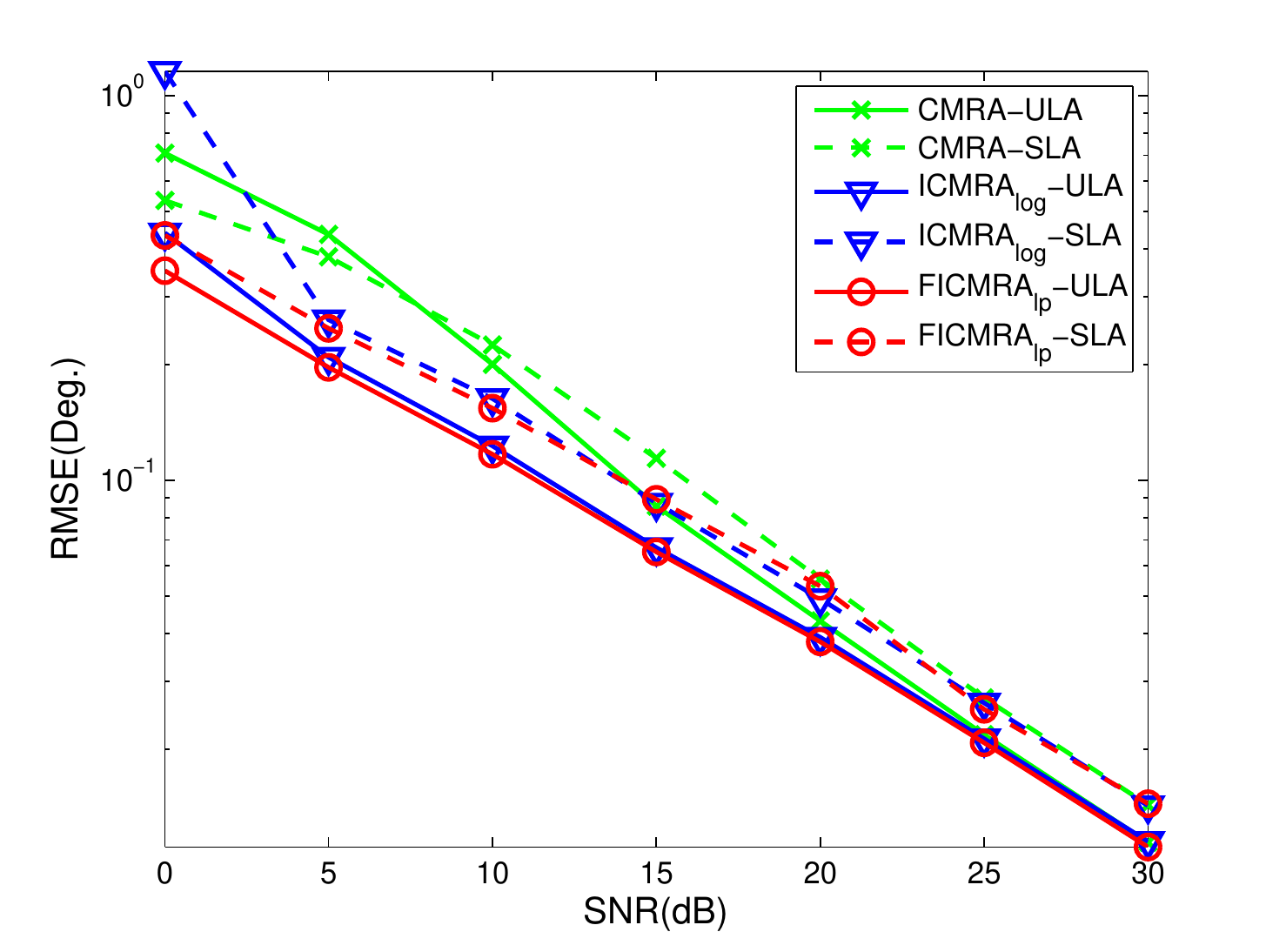}
\caption{RMSE results of CMRA, $\text{ICMRA}_{\text{log}}$ and $\text{FICMRA}_{\ell_p}$ in the cases of a $7$-element ULA and a $4$-element SLA with $\bm{\Omega} = \{1,2,5,7\}$. The signals impinge onto the arrays from directions of $[-5^{\circ},5^{\circ}]$, $L=200$.}
\label{fig RMSE SNR ULA and SLA}
\end{figure}

We now evaluate the DOA estimation performance of the proposed methods with comparison to CMRA. Assume two narrowband signals impinge onto a $7$-element ULA from directions of $[-1^{\circ},3^{\circ}]$. We collect $L=400$ snapshots which are corrupted by i.i.d. Gaussian noise of unit variance with $\text{SNR}=10$dB. We show the DOA estimation results of the compared methods in Fig. \ref{fig DOA estimation comparison}. The powers can be estimated by solving a least squares problem given the corresponding DOA estimates. It can be seen that the outputs of CMRA are biased while the proposed methods ICMRAs/FICMRAs are able to locate the spatially adjacent sources with a higher precision. It is interesting to note that the two sources tend to merge together for all of these methods. It is because the two DOAs are very closely-located. Specifically, when the noise is heavier, or the number of snapshots is smaller, or the sources become closer, we will find that the two red point paths completely merge together, which indicates that we can detect only one single source at this region in some cases. More simulations show that with a higher SNR, larger $L$ and separation, the two paths tend to be parallel, which means the DOAs are estimated with a higher precision. Finally, the running times of each method are compared in Table \ref{table CPUTime}.
It can be seen that the ICMRAs require the most computations as expected. Note that the CPU times of $\text{ICMRA}_{\text{log}}$ and $\text{ICMRA}_{\ell_p}$ are more than $20$ (i.e., the maximum number of iterations) times larger than that of CMRA. This is because the ICMRA involves an additional SVD in each iteration, which is also time-consuming. As mentioned in Section \ref{sec reweighted CMRA B}, $\text{ICMRA}_{\text{log}}$ and $\text{ICMRA}_{\ell_p}$ can be accelerated since the SVD is not required. Finally, it can be observed that the FICMRAs are computationally very efficient as compared to ICMRAs as well as CMRA.

Next, we study the success rates of ICMRAs, FICMRAs and CMRA in terms of the angle separation $\Delta\theta$ and the number of sensors $M$. It is assumed that two signals impinge onto an $M$-element ULA from directions of $[-1^{\circ},-1^{\circ}+\Delta\theta]$ where $\Delta\theta$ varies from $4^{\circ}$ to $14^{\circ}$ and the number of sensors increases from $4$ to $14$. Successful estimate is declared if the root mean square error (RMSE) of DOA estimation is less than $0.1$. The success rates for each pair of $M$ and $\Delta\theta$ are presented based on $50$ independent trials in Fig. \ref{fig grayscale} where white means complete success and black means complete failure. It can be seen that a large value of $M$ or $\Delta\theta$ leads to a high success rate while a small value of $M$ or $\Delta\theta$ may result in unsatisfactory estimates, leading to a phase transition in the $M$-$\Delta\theta$ domain. By comparing Fig. \ref{fig grayscale}(a)-(g), we can see that ICMRAs/FICMRAs with the three penalties enlarge the success region compared to CMRA. In particular, by noting the fact that $\text{ICMRA}_{\text{log}}$ is able to correctly locate the signals separated by $\Delta\theta=4^{\circ}$ with a $9$-element ULA where $\text{ICMRA}_{\ell_p}$ and $\text{ICMRA}_{\text{lap}}$ fail, it can be concluded that the Logarithm penalty is superior to the $\ell_p$ and Laplace penalties. Furthermore, Fig. \ref{fig grayscale}(b)-(g) also show that the performance improvement has mostly appeared in the small $\Delta\theta$s region, demonstrating the ability of our methods in enhancing resolution.

\begin{figure}[!t]
\begin{minipage}[b]{.48\linewidth}
  \centering
  \centerline{\includegraphics[width=1.7in]{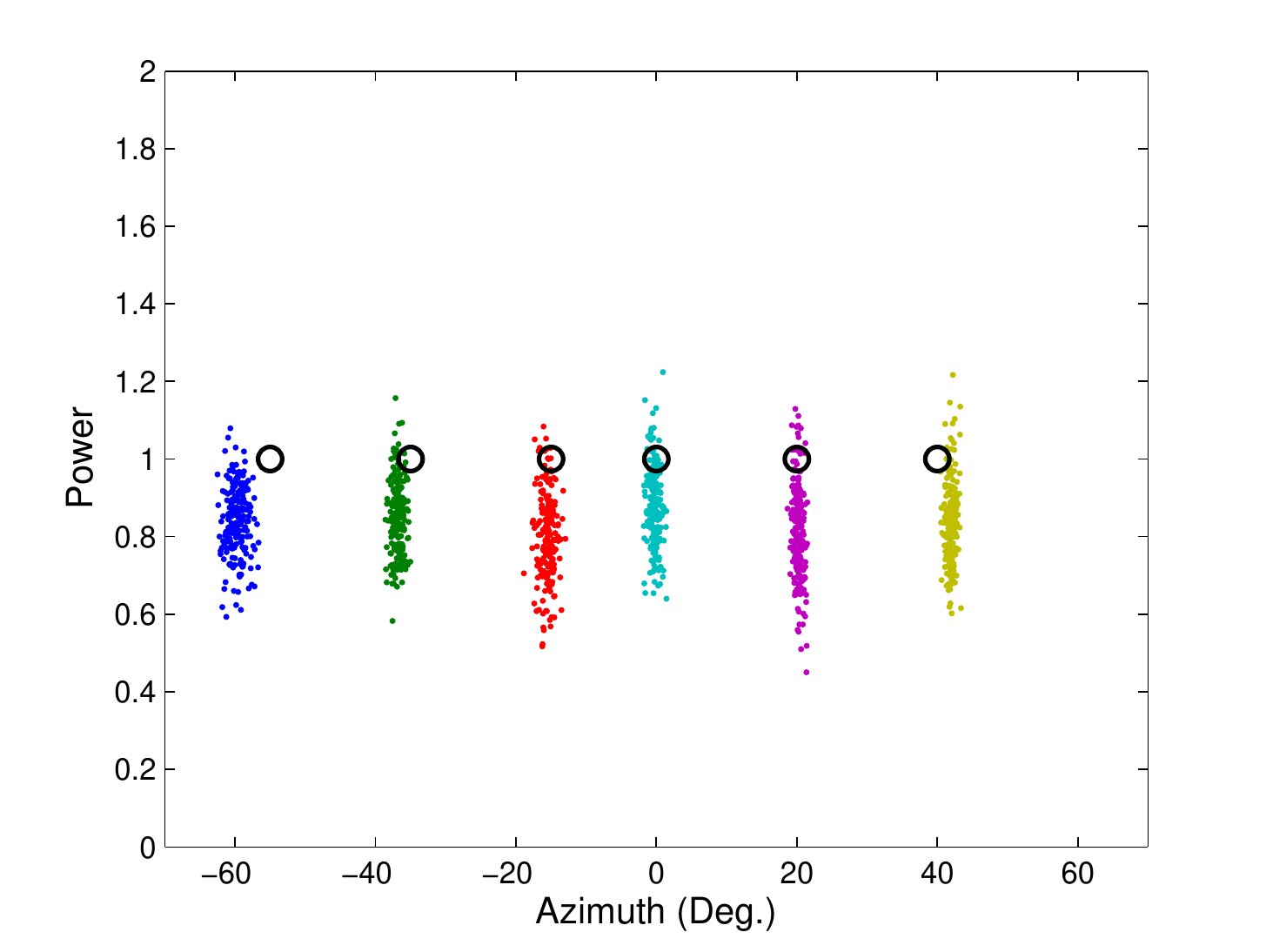}}
  \centerline{(a) CMRA}\medskip
\end{minipage}
\hfill
\begin{minipage}[b]{0.48\linewidth}
  \centering
  \centerline{\includegraphics[width=1.7in]{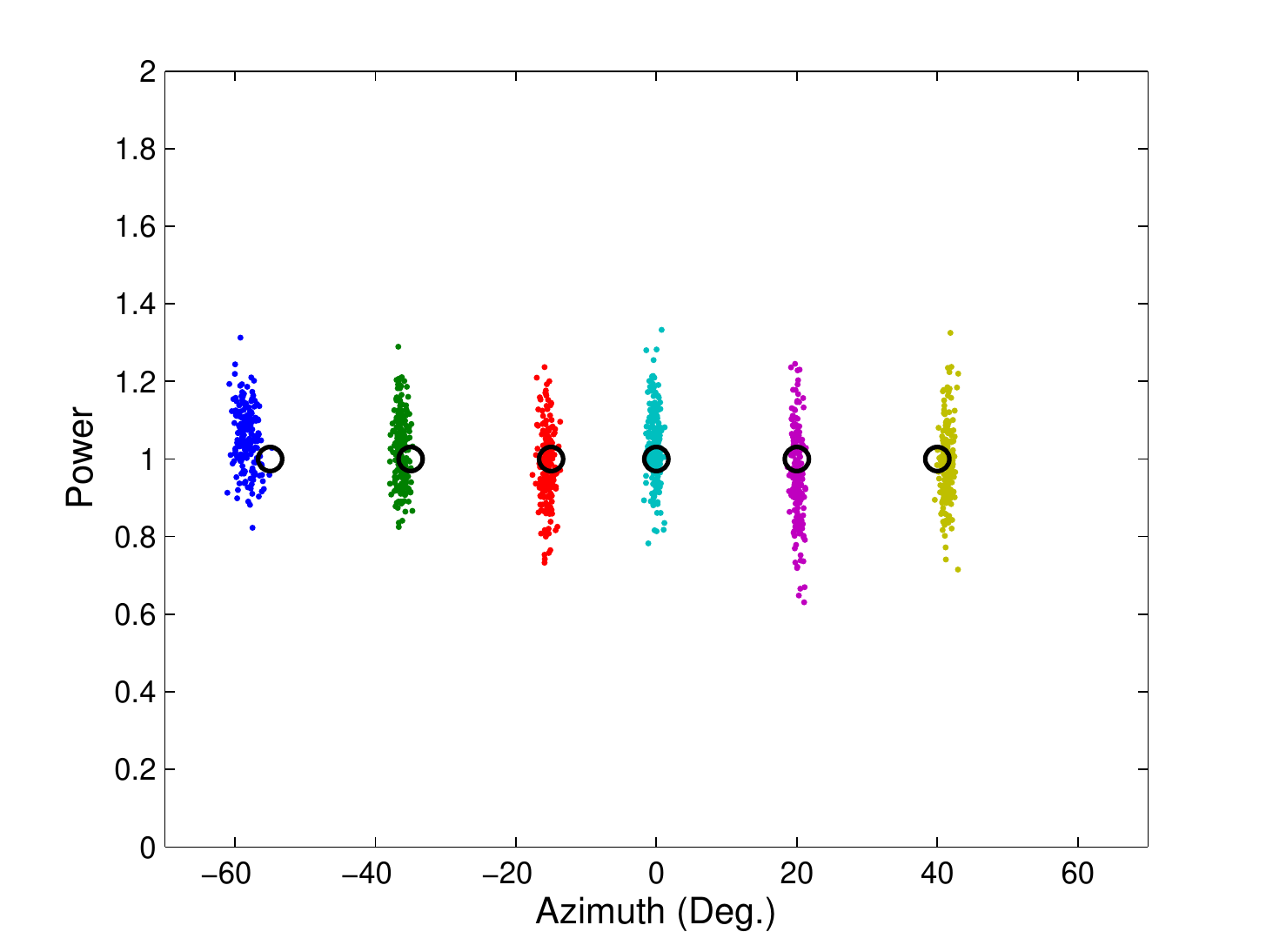}}
  \centerline{(b) $\text{ICMRA}_{\ell_p}$}\medskip
\end{minipage}
\caption{DOA estimation comparisons of CMRA, $\text{ICMRA}_{\ell_p}$ with more sources than sensors.}
\label{fig 6 sources}
\end{figure}

\begin{figure*}[!t]
\centerline{\subfloat[RMSE Comparisons]{\includegraphics[width=3in]{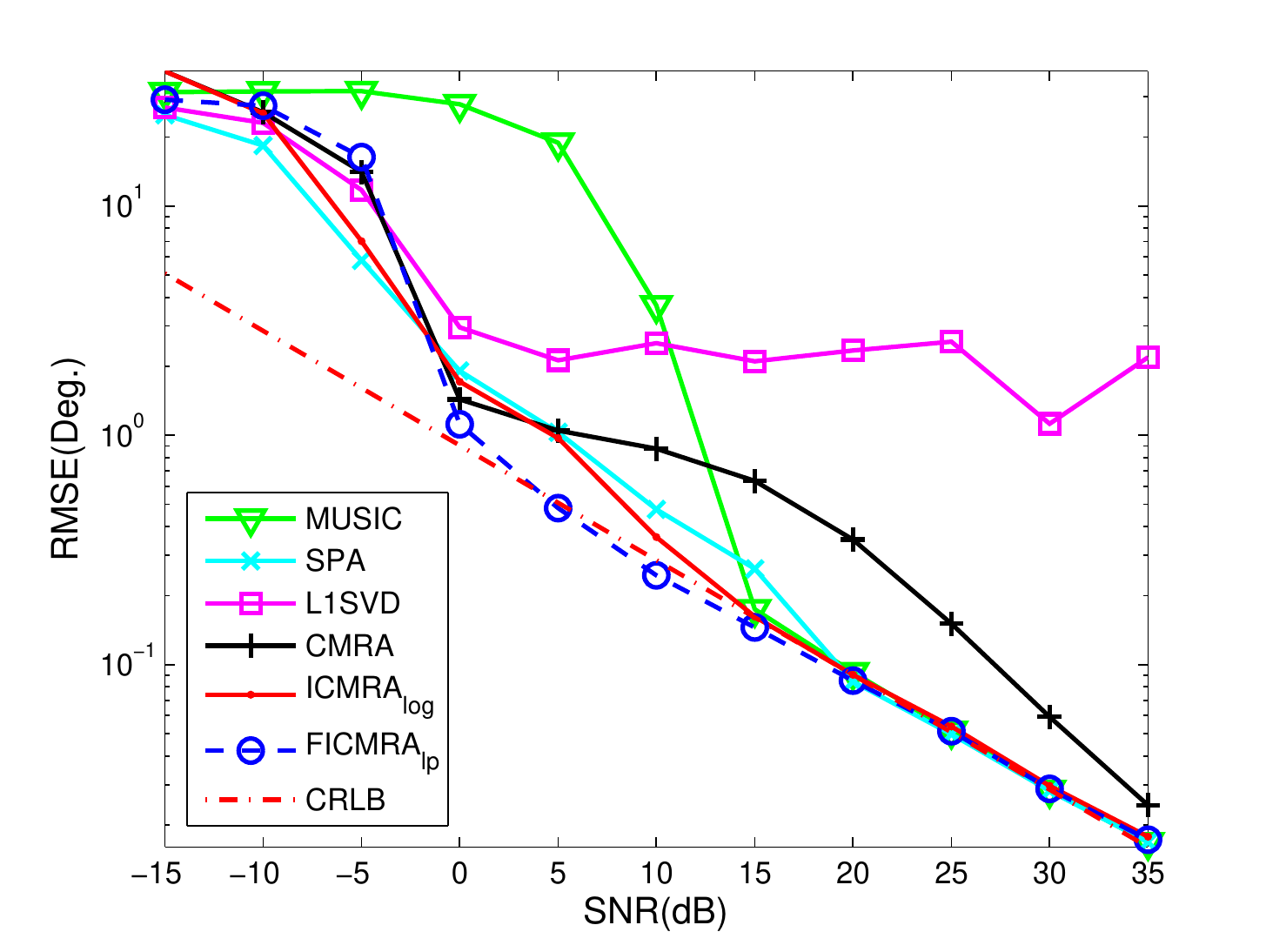}}
\hfil
\subfloat[CPU Time Comparisons]{\includegraphics[width=3in]{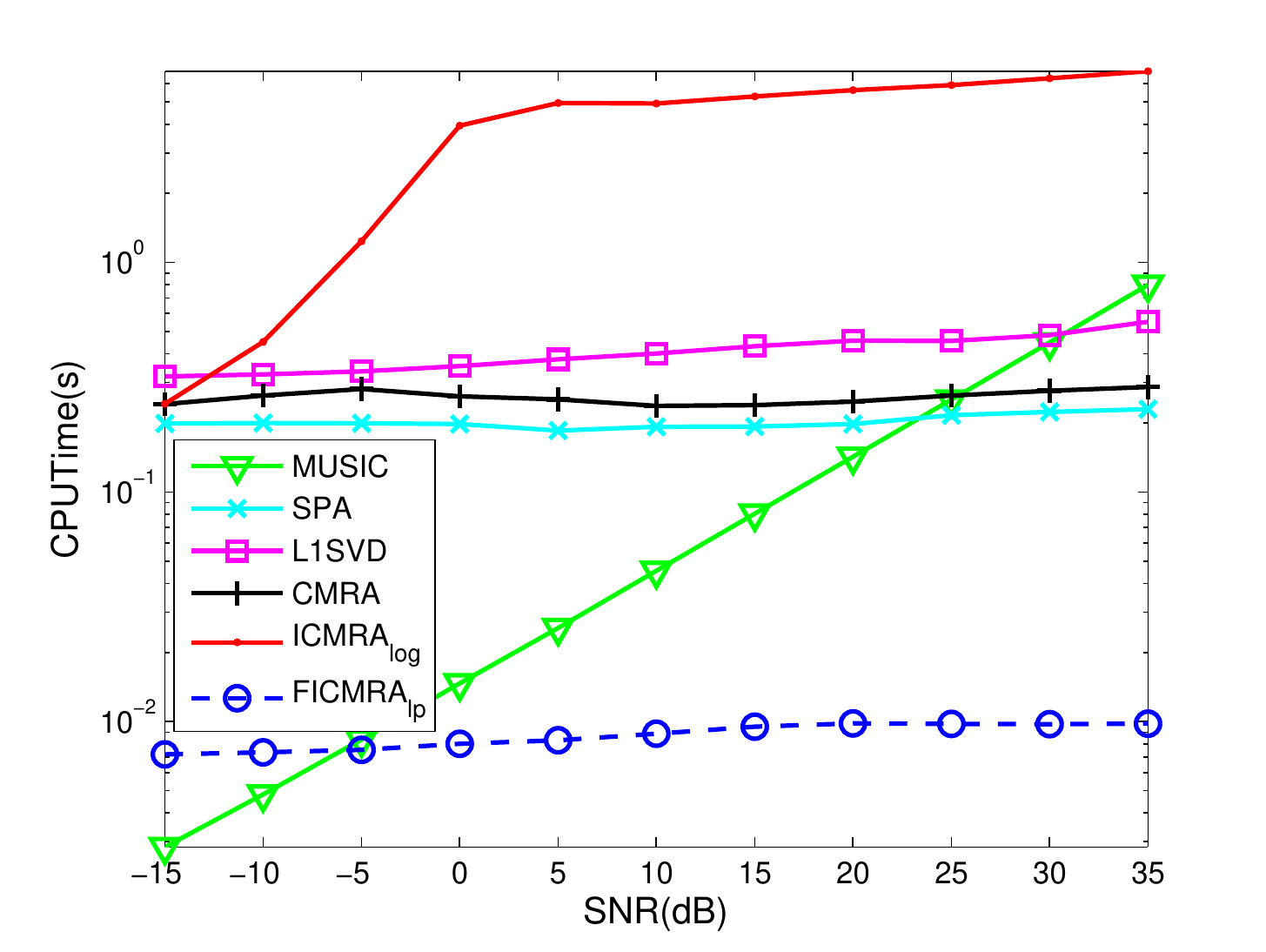}}}
\caption{Performance comparison of $\text{ICMRA}_{\text{log}}$, $\text{FICMRA}_{\text{log}}$ and other methods for two uncorrelated sources impinging from $[-1^{\circ}+v, 3^{\circ}+v]$ with the SNR varying from $-15$dB to $35$dB, $L=200$.}
\label{Fig RMSE_SNR}
\end{figure*}

We then consider the estimation performance of the proposed methods in a $4$-element SLA case with $\bm{\Omega}=\{1,2,5,7\}$, as well as a ULA case with the same aperture size, i.e., a $7$-element ULA. Assume two uncorrelated sources impinge onto both arrays from directions of $[-5^{\circ},5^{\circ}]$ and $200$ snapshots are collected, simultaneously. We run $400$ independent trials and show the RMSE of CMRA and our methods with the SNR varying from $0$dB to $30$dB in Fig. \ref{fig RMSE SNR ULA and SLA}. For brevity, in this example, we only consider the $\text{ICMRA}_{\text{log}}$ and $\text{FICMRA}_{\ell_p}$ as the two representative methods of ICMRA and FICMRA, respectively. From Fig. \ref{fig RMSE SNR ULA and SLA} it can be seen that our proposed methods give a better performance than CMRA in both ULA and SLA cases with respect to most compared region. Furthermore, it should be noted that, compared to the ULA case, the SLA case is able to save $43\%$ data with only a small degradation of $1$dB to $1.5$dB in terms of RMSE for the three methods when SNR$\geq5$dB. Finally, it can be seen that for each array, the curves of the three methods tend to merge together when SNR is larger than $25$dB.

We finally evaluate the performance of our methods in the SLA case with more sources than sensors. In particular, assume six uncorrelated sources impinge onto a $4$-element SLA with $\bm{\Omega} = \{ 1,2,5,7 \}$ from directions of $[-55^{\circ}, -35^{\circ}, -15^{\circ}, 0^{\circ}, 20^{\circ}, 40^{\circ}]$ simultaneously. $400$ snapshots which are corrupted by i.i.d. Gaussian noise are collected for source localization and the $\text{SNR}$ is set to $0\text{dB}$. We compare the estimation performance of CMRA and $\text{ICMRA}_{\ell_p}$ via $200$ independent trials and show the results in Fig. \ref{fig 6 sources}, from which it can be seen that both methods can correctly locate all the $6$ sources while $\text{ICMRA}_{\ell_p}$ is able to show better performance in terms of both DOAs and powers. To be specific, the average RMSEs of CMRA with respect to DOA and power are $0.6395$ and $0.1828$, respectively, while those of $\text{ICMRA}_{\ell_p}$ are $0.5486$ and $0.1331$, respectively.
We have also carried out this experiment on a $4$-element SLA using the Logarithm and Laplace penalties, where the Logarithm penalty failed to provide a satisfactory performance in all trials since several outliers are observed in the experiment, especially when the number of sources is equal to or larger than that of sensors. In contrast, the proposed method with the Laplace penalty is able to give a good performance. Here, we report this observation and suggest to use $\text{(F)ICMRA}_{\ell_p}$ or $\text{(F)ICMRA}_{\text{lap}}$ in the SLA case with the number of sources being larger than that of sensors and use $\text{(F)ICMRA}_{\text{log}}$ in other cases.

\subsection{Comparisons with Prior Arts}

In this section, the performance of the proposed methods (F)ICMRA with the Logarithm penalty is studied with comparison to other existing methods.  
The prior arts considered in this part are CMRA, MUSIC, L1-SVD and SPA. For the sparsity-based method L1-SVD, the discretized interval is set to $2^{\circ}$ and the iterative grid refinement procedure is employed \cite{malioutov2005sparse}. In L1-SVD the number of iterations to refine the grid is set to $5$ and the grid interval of the last iteration is set to $0.1\times10^{-\frac{\text{SNR}}{20}}$ for accuracy consideration. As a consequence, the RMSE of L1-SVD can achieve the Crammer-Rao Lower Bound (CRLB) theoretically.

\begin{figure}[!t]
\centering
\includegraphics[width=3in]{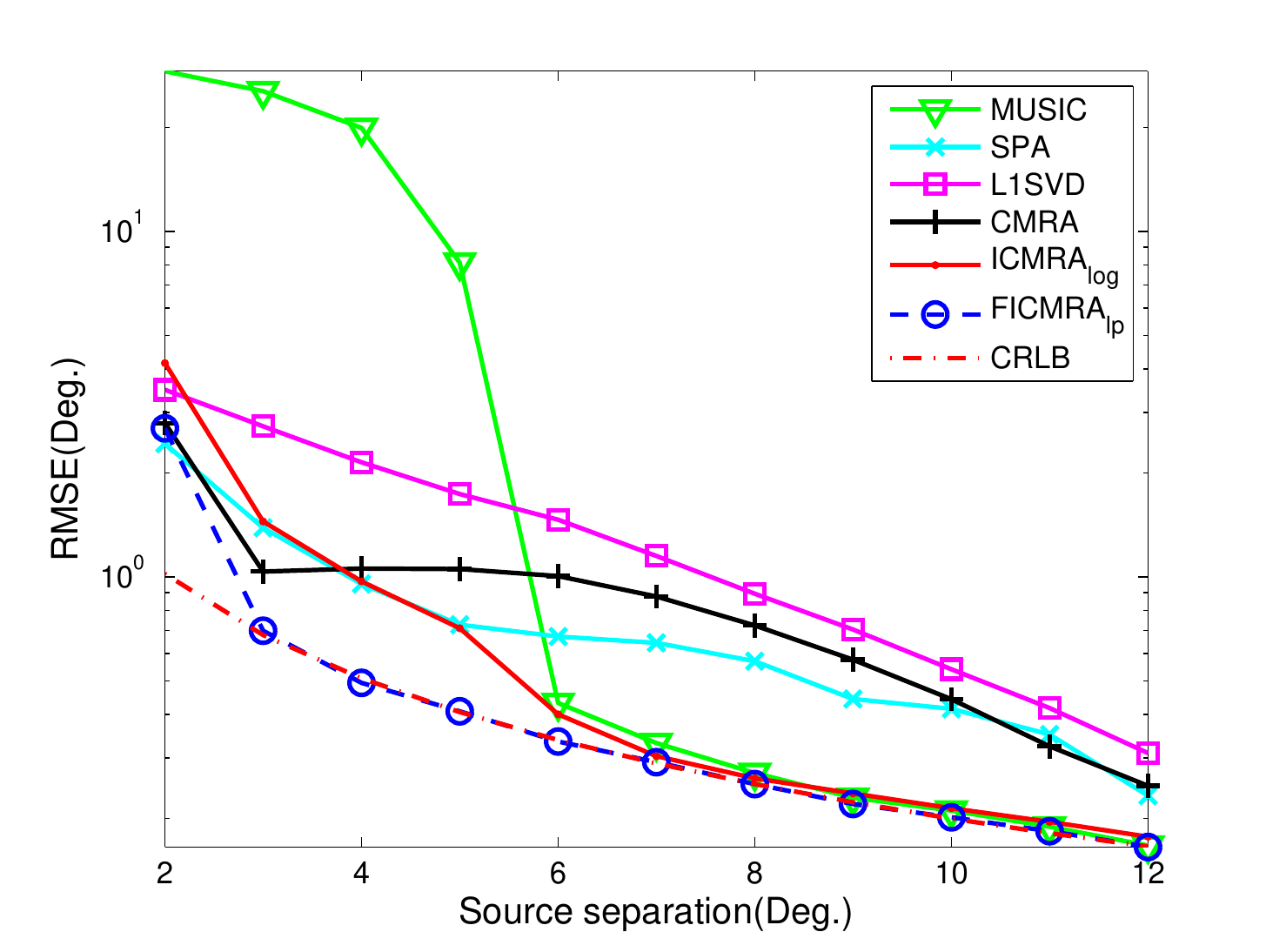}
\caption{Performance comparison of $\text{ICMRA}_{\text{log}}$, $\text{FICMRA}_{\text{log}}$ and other methods for two uncorrelated sources impinging from $[0^{\circ}, \Delta\theta]$ with $\Delta\theta$ varying from $2^{\circ}$ to $12^{\circ}$. We set SNR$=15$dB and $L=200$.}
\label{fig RMSE delta theta}
\end{figure}

\begin{figure*}[!t]
\centerline{\subfloat[RMSE Comparisons]{\includegraphics[width=3in]{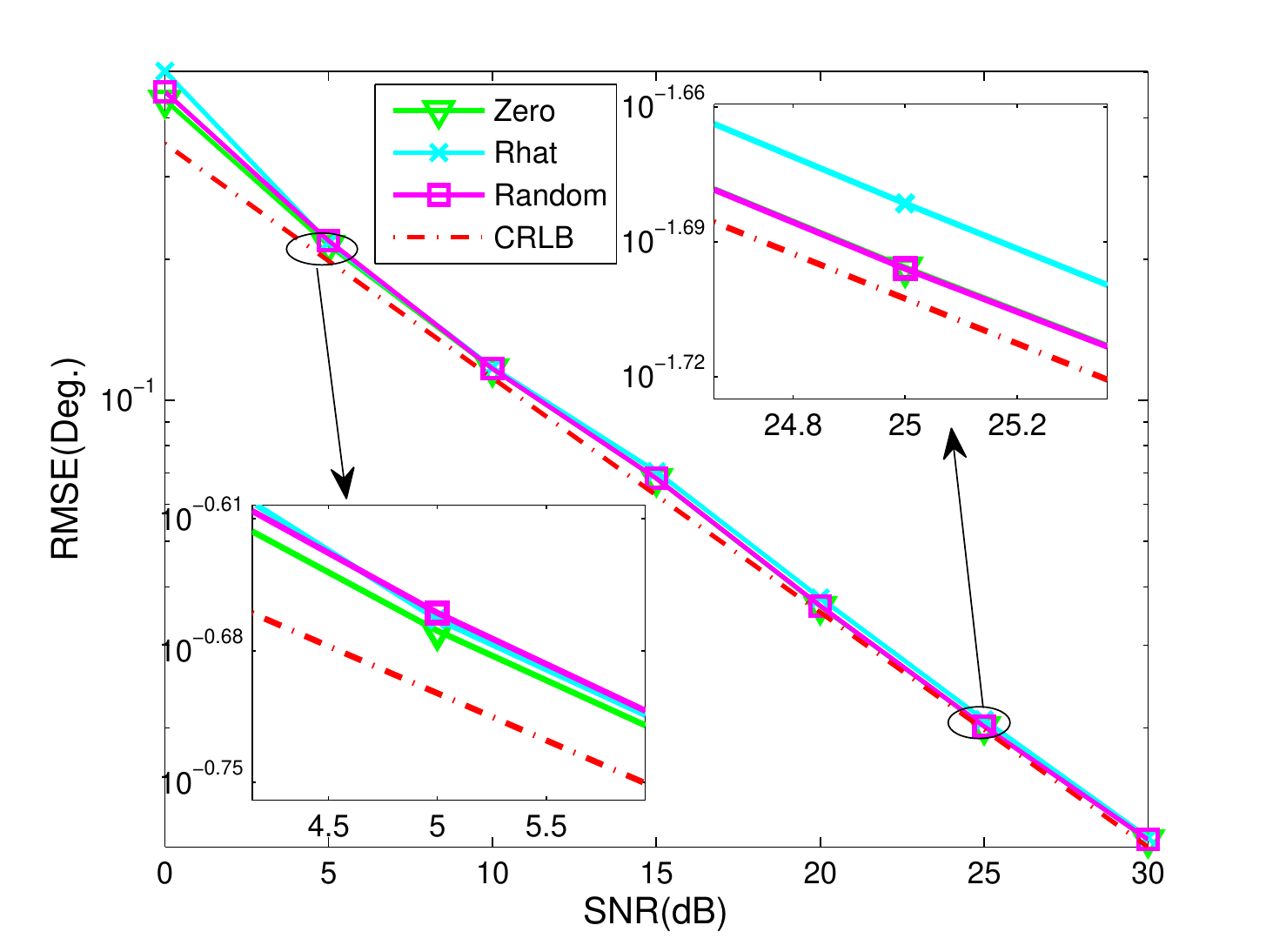}}
\hfil
\subfloat[CPU Time Comparisons]{\includegraphics[width=3in]{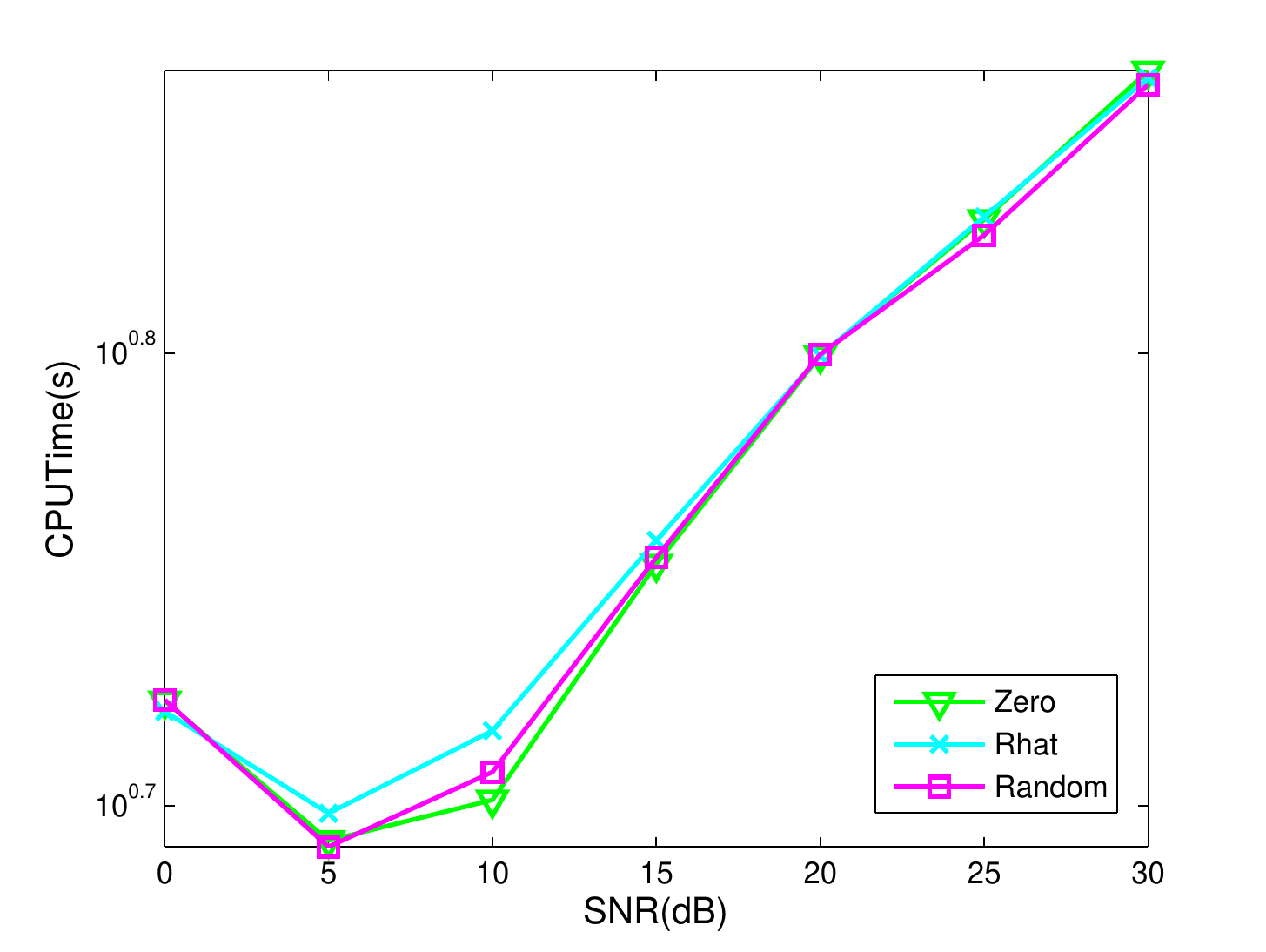}}}
\caption{Performance comparison of $\text{ICMRA}_{\text{log}}$ with respect to different initializations for two uncorrelated sources impinging from $[-5^{\circ}, 5^{\circ}]$ with the SNR varying from $0$dB to $30$dB, $L=200$.}
\label{Fig RMSE_SNR_diffInitials}
\end{figure*}

Suppose that two uncorrelated signals impinge onto a $7$-element ULA from $[-1^{\circ}+v, 3^{\circ}+v]$ where $v$ is chosen randomly and uniformly within $[-1^{\circ}, 1^{\circ}]$ in each trial. Note that the two sources are spatially adjacent since they are only separated by around $4^{\circ}$.
Assume that $200$ snapshots are collected in each trial for DOA estimation. We compare the statistical results of the aforementioned methods obtained from $400$ independent trials in Fig. \ref{Fig RMSE_SNR} with the SNR varying from $-15$dB to $35$dB. Fig. \ref{Fig RMSE_SNR}(a) shows the RMSE comparisons of these methods, from which we can see that $\text{FICMRA}_{\text{log}}$ gives the best performance when $\text{SNR}\geq 0\text{dB}$ and approaches the CRLB curve first. $\text{ICMRA}_{\text{log}}$ shows a better result than CMRA in most compared region. In particular, the curve of CMRA deviates the CRLB curve in the middle range of SNR (it approaches the CRLB when SNR$>40$dB) while $\text{ICMRA}_{\text{log}}$ revises this deviation and is able to coincide with the CRLB curve when SNR$\geq 15$dB. In comparison, MUSIC and SPA approach the CRLB curve when SNR$\geq 15$dB and $20$dB, respectively. L1-SVD is unable to reach the CRLB in this case of closely-located sources (one possible reason may be the high correlation between the columns of the manifold matrix in the sparsity-based methods \cite{liuzhangmeng2012SBL}).
The running times of each methods are also compared in Fig. \ref{Fig RMSE_SNR}(b). $\text{ICMRA}_{\text{log}}$ is time-consuming because of the iterative procedure. Nevertheless, it should be noted that, the increased computational cost of ICMRA can be regarded as a sacrifice for resolution improvement. The running time of L1-SVD, CMRA and SPA are comparable with each other and nearly remain stable in the compared SNR region. The $\text{FICMRA}_{\text{log}}$ requires much smaller computations than other methods when SNR$\geq -5$dB. It should be noted that the effectiveness of $\text{FICMRA}_{\text{log}}$ is because of the derived closed-form solution of model (\ref{eq Lagrangian of ICMRA}) rather than at any cost of the performance deterioration. On the contrary, we can observe that, compared to $\text{ICMRA}_{\text{log}}$, $\text{FICMRA}_{\text{log}}$ is able to reduce the RMSE in some extent. 
Finally, the running time of MUSIC increases exponentially since the sampling grid size grows as the SNR increases.

We then evaluate the RMSE with respect to different source separations $\Delta\theta$. It is assumed two sources impinge onto a $7$-element ULA from $[0^{\circ},\Delta\theta]$ and $200$ snapshots are collected. We set SNR$=15$dB, $L = 200$ and show the RMSEs of these methods with $\Delta\theta$ varying from $2^{\circ}$ to $12^{\circ}$ in Fig. \ref{fig RMSE delta theta}, which reveals the superiority of our proposed methods in separating spatially adjacent signals. It should be mentioned that, for $\text{ICMRA}_{\text{log}}$ in this simulation, we set $\delta=1.1$ as an initial value when $\Delta\theta$ starts from $2^{\circ}$ and gradually increase $\delta$ until $2$ with step $0.1$ as $\Delta\theta$ increases. From Fig. \ref{fig RMSE delta theta} it can be observed that our proposed methods outperform other methods in most cases. In particular, $\text{FICMRA}_{\text{log}}$ is able to coincide with the CRLB when the sources are only $3^{\circ}$ apart. $\text{ICMRA}_{\text{log}}$ and MUSIC approach the CRLB when the source separation is larger than $6^{\circ}$. Note that when $3^{\circ}<\Delta\theta<6^{\circ}$, $\text{ICMRA}_{\text{log}}$ still can identify the two adjacent sources while MUSIC fails to do so. In contrast, L1-SVD, CMRA and SPA are far from the CRLB in the compared region.

\subsection{Influence of Initialization}
\label{sec initialization}

In the above simulations we have set $\bm{u}_0 = \bm{0}$ so that the first iteration is equivalent to the CMRA method. The simulation results reveal that this initialization is able to give a satisfactory performance. In this section, we aim to show the superiority of our initialization by emprically exploring the performance of our method with respect to different initializations.

The experiment is carried out on a 7-element ULA, and the parameter setting of this experiment is similar to that of Fig. \ref{fig RMSE SNR ULA and SLA} but with the following different initializations of $\bm{u}_0$: 1) zero vector, as we have recommended above, 2) random vector with Gaussian distribution, 3) the first column of $\hat{\bm{R}}^{\text{full}}$, which can be calculated according to Remark \ref{Remark Rfull}. We choose the $\text{ICMRA}_{\text{log}}$ for simulation and show the results in Fig. \ref{Fig RMSE_SNR_diffInitials} with the SNR varying from $0$dB to $30$dB. It can be observed that the performance of $\text{ICMRA}_{\text{log}}$ with respect to different initializations show similar performance in terms of both RMSE and CPU time. In particular,  Fig. \ref{Fig RMSE_SNR_diffInitials}(b) and the subfigures in Fig. \ref{Fig RMSE_SNR_diffInitials}(a) indicate that the zero vector initialization that we employed is slightly better that the other two initializations, especially in the middle SNR range. Finally, it can be concluded that our method $\text{ICMRA}_{\text{log}}$ using the iterative procedure is insensitive to the initialization. The simulation results with respect to other penalties and FICMRAs reveal the same conclusion and are omitted.

\section{Conclusions}
\label{sec conclusions}

In this paper, we have proposed a reweighted method named ICMRA by applying a family of nonconvex penalties on the singular values of the covariance matrix as the sparsity metrics. We have also given the convergence analysis of the proposed method as well as two fast implementation algorithms. We have shown that ICMRA can be considered as a sparsity-based method with infinite number of sampling grids. The proposed ICMRA can enhance sparsity and resolution compared to CMRA, as verified through simulations. In our future studies, we will extend the proposed method into the case of generalized arrays. We will be specifically interested in the case where the sensors of an array can be arbitrarily located without being restricted to the half-wavelength sensor spacing, which is the limitation of the subspace-based methods and the SCMR methods.

\appendices

\section{Proof of Theorem \ref{theorem convergent}}
\label{appendix convergent}

To prove Theorem \ref{theorem convergent}, we introduce the following lemma.
\begin{lemma}[\cite{ICRA2014Mohammadi}]\label{lemma G is concave}
  Assume $\mathcal{G}^{\epsilon}[T(\bm{u})]=\sum_i g^{\epsilon}[\sigma_i[T(\bm{u})]]$. If $g^{\epsilon}$ is twice differentiable, strictly concave, and $g^{\epsilon''}(x) \leq -m_u<0$ for any bounded $u$ and $0\leq x\leq u$, then $\mathcal{G}^{\epsilon}[T(\bm{u})]$ is strictly concave, and for any bounded $\bm{u}, \bm{v}\in \mathbb{R}^N, \bm{u}\neq \bm{v}$, there is some $m>0$ such that
  \begin{equation}\label{eq lemma}
  \begin{split}
    \mathcal{G}^{\epsilon}[T(\bm{u})] - \mathcal{G}^{\epsilon}[T(\bm{v})] \leq &\big\langle T(\bm{u}-\bm{v}), \nabla\mathcal{G}^{\epsilon}[T(\bm{v})]\rangle\\
    &- \frac{m}{2}\big\| T(\bm{u}-\bm{v}) \big\|_F^2.
  \end{split}
  \end{equation}
\end{lemma}
We first show that $\mathcal{G}^{\epsilon}[T(\bm{u}_j)]$ is convergent.
According to Lemma \ref{lemma G is concave} and the concavity of $\mathcal{G}^{\epsilon}[T(\bm{u})]$, we have
\begin{equation}\label{eq11 in theorem}
\begin{split}
  \mathcal{G}^{\epsilon}[T(\bm{u}_{j+1})] - \mathcal{G}^{\epsilon}[T(\bm{u}_{j})] \leq &\text{tr}\big[\nabla \mathcal{G}^{\epsilon}[T(\bm{u}_j)]\big( T(\bm{u}_{j+1}-\bm{u}_{j}) \big)\big]\\
  &- \frac{m}{2}\big\| T(\bm{u}_{j+1}-\bm{u}_j) \big\|_F^2,
\end{split}
\end{equation}
which directly results in,
\begin{equation}\label{eq1 in theorem}
  \mathcal{G}^{\epsilon}[T(\bm{u}_j)] - \mathcal{G}^{\epsilon}[T(\bm{u}_{j+1})] \geq \text{tr}\big[\nabla \mathcal{G}^{\epsilon}[T(\bm{u}_j)]\big( T(\bm{u}_j-\bm{u}_{j+1}) \big)\big].
\end{equation}
Then, since $\bm{u}_{j+1}$ is updated according to model (\ref{eq RCMRA4}), we can conclude that,
\begin{equation}\label{eq 2 in theorem}
  \text{tr}[\nabla \mathcal{G}^{\epsilon}[T(\bm{u}_j)]T(\bm{u}_{j+1})] \leq \text{tr}[\nabla \mathcal{G}^{\epsilon}[T(\bm{u}_j)]T(\bm{u}_{j})],
\end{equation}
which together with (\ref{eq1 in theorem}) confirms that
\begin{equation}
  \mathcal{G}^{\epsilon}[T(\bm{u}_j)] - \mathcal{G}^{\epsilon}[T(\bm{u}_{j+1})] \geq 0.
\end{equation}
Further, because $\mathcal{G}^{\epsilon}[T(\bm{u})]=\sum_i g^{\epsilon}\big[\sigma_i[T(\bm{u})]\big]\geq 0$, the first property can be drawn.

To prove the second property, we first show that the sequence $\{\bm{u}_j\}$ is convergent. We start by applying Lemma \ref{lemma G is concave} on $\mathcal{G}^{\epsilon}[T(\bm{u})]$ to get,
\begin{equation}\label{eq3 in theorem}
\begin{split}
  &\mathcal{G}^{\epsilon}[T(\bm{u}_j)] - \mathcal{G}^{\epsilon}[T(\bm{u}_{j+1})] \\
  \geq &\text{tr}\big[\nabla \mathcal{G}^{\epsilon}[T(\bm{u}_j)]\big( T(\bm{u}_j-\bm{u}_{j+1}) \big)\big] + \frac{m}{2} \big\|T(\bm{u}_{j+1}-\bm{u}_j)\big\|_F^2\\
  \geq &\frac{m}{2} \big\|T(\bm{u}_{j+1}-\bm{u}_j)\big\|_F^2.
\end{split}
\end{equation}
Summing the inequality above for all $j\geq 0$ gives
\begin{equation}
  \frac{m}{2} \sum_{j=0}^{+\infty}\big\| T(\bm{u}_{j+1}-\bm{u}_j) \big\|_F^2 \leq \mathcal{G}^{\epsilon}[T(\bm{u}_0)],
\end{equation}
which implies that $\lim_{j\rightarrow +\infty} T(\bm{u}_{j+1}-\bm{u}_j) = \bm{0}$. Hence the sequence $\{\bm{u}_j\}$ is convergent.

To prove that $\{\bm{u}_j\}$ converges to a local minimum, we first note that when it converges, i.e., $\bm{u}_j \rightarrow \bm{u}^*$ as $j\rightarrow \infty$, the models (\ref{eq RCMRA2}) and (\ref{eq RCMRA4}) have the same KKT conditions. This is because the constraints and the derivatives of the objective functions in the two models are the same when the MM method converges.
Further, since the objective function in (\ref{eq RCMRA2}) monotonically decreases \cite{ICRA2014Mohammadi}, we can confirm that $\bm{u}^*$ is a local minimum of (\ref{eq RCMRA2}). This completes the proof.

\ifCLASSOPTIONcaptionsoff
  \newpage
\fi

\bibliographystyle{IEEEtran}

\end{document}